\documentclass{article}
\usepackage[top=1in, bottom=1in, left=1in, right=1in]{geometry}
\usepackage{amsmath,amsthm,amssymb,color,latexsym,dsfont,mathtools} 
\usepackage{graphicx, float}
\usepackage{braket}
\usepackage[hypertexnames=false]{hyperref}
\usepackage{cleveref}
\usepackage{bm}
\usepackage[doi=true, isbn=false, eprint=false, backend=biber, citestyle=nature, sorting=none, giveninits=true, url=false]{biblatex}
\usepackage{authblk}
\usepackage{lipsum}
\usepackage[acronym]{glossaries}
\usepackage{makecell}
\usepackage{multicol}
\usepackage[dvipsnames]{xcolor}
\usepackage{comment}
\renewcommand{\vec}[1]{\ensuremath{\boldsymbol{#1}}}

\newtheorem*{theorem*}{Theorem}
\newtheorem*{corollary*}{Corollary}
\newtheorem{theorem}{Theorem}
\newtheorem{definition}{Definition}
\newtheorem{proposition}{Proposition}
\newtheorem{lemma}{Lemma}
\newtheorem{corollary}{Corollary}

\AddToHook{cmd/appendix/before}{
    \crefalias{section}{appendix}
    \crefalias{subsection}{appendix}
}

\newacronym{VQA}{VQA}{variational quantum algorithm}
\newacronym{VQE}{VQE}{variational quantum eigensolver}

\title{On the convergence of the variational quantum eigensolver and quantum optimal control}

\author[1]{Marco Wiedmann\footnote{E-Mail: marco.wiedmann@fau.de}}
\author[1]{Daniel Burgarth}
\author[2]{Gunther Dirr}
\author[3,4]{Thomas Schulte-Herbrüggen}
\author[3,4]{Emanuel Malvetti}
\author[5]{Christian Arenz\footnote{E-Mail: carenz1@asu.edu}}
\affil[1]{Department Physik, Friedrich-Alexander-Universität Erlangen-Nürnberg, 91058 Erlangen, Germany}
\affil[2]{Institute of Mathematics, Julius-Maximilians-Universität Würzburg, 97074 Würzburg, Germany}
\affil[3]{School of Natural Sciences, Technische Universität München, 85748 Garching, Germany}
\affil[4]{Munich Center for Quantum Science and Technology (MCQST) \& Munich Quantum Valley (MQV), 80799 München, Germany}
\affil[5]{School of Electrical, Computer and Energy Engineering, Arizona State University, Tempe, AZ 85281, USA}

\begin{document}
    \maketitle
    
\vspace{-0.75cm}

\begin{abstract}
When do variational quantum algorithms converge to a globally optimal solution? Despite extensive work, this question remains largely open. We develop a convergence theory for the variational quantum eigensolver (VQE), using quantum control landscape terminology to prove a sufficient criterion guaranteeing convergence to a Hamiltonian's ground state for almost all initial parameters. Specifically, we show that if (i) the parameterized unitary transformation allows movement in all tangent-space directions in a bounded manner (local surjectivity), and (ii) the gradient descent terminates, then VQE converges almost surely to a ground state. Under similar assumptions, we also guarantee almost-sure convergence to a global optimum on a smaller, closed unitary Lie subgroup, with representation theory providing a criterion for when this optimum is a ground state. We construct examples satisfying condition (i) and analyze two common circuit ansatz families, then discuss regularization techniques that ensure gradient descent terminates and connect to the halting problem.
\end{abstract}

\begin{multicols}{2}

\section{Introduction}
The advent of noisy, intermediate-scale quantum devices has sparked an explosion in the development of \glspl{VQA} \cite{cerezo_variational_2021}. These algorithms aim to solve an optimization problem by using classical and quantum computing resources in tandem. Namely, a classical optimizer is employed in conjunction with a quantum device to optimize parameters of a quantum circuit to solve a given computational problem. The quantum device is used to assess the quality of the proposed solution at each step while the classical computer provides the parameter update. 

VQAs face the challenge that the classical optimization routine often has to navigate a complicated optimization landscape \cite{VQA_traps, PhysRevLett.127.120502}. The existence of local optima can hinder the search for the optimal solution, which makes it challenging to develop rigorous criteria for when VQAs solve the optimization problem at hand.
It is known that typically, in the regime of few parameters compared to the problem dimension, these local optima are far away from the global optimum, meaning that in many cases the solution found by VQAs is of bad quality \cite{VQA_traps, anschuetz_critical_2021}.
At the heart of this challenge lies the nonconvex nature of the optimization problem and the complex interplay between the parameterized quantum circuit and the problem instance \cite{lee_progress_2021}. To date, a variety of numerical results from previous works have shown that overparameterization can improve convergence to the globally optimal solution \cite{kiani_learning_2020, PRXQuantum.1.020319,larocca_theory_2023,lee_progress_2021,you_convergence_2022}.
However, rigorous criteria that characterize when a parameterized quantum circuit yields convergence of the VQA to the optimal solution remain challenging to derive \cite{you_convergence_2022}. Here we address this challenge by building on the long history of quantum control landscapes \cite{Chakrabarti01102007,GE2022314} to develop a convergence theory for a specific class of VQAs: the \gls{VQE} \cite{TILLY20221} that aims to minimize the expectation value of a given Hamiltonian, and thus approximates its ground state, on a quantum device.

VQAs aim to optimize a cost function $J[U]$ that characterizes how close a quantum circuit $U$ is to the problem solution. If the minimization is carried out directly over $U$ (i.e.~over the unitary group), the rich theory of Riemannian gradient flows, cf. e.g. \cite{Bro88+91, HM94, smith_optimization_1994, SGDH08, absil_optimization_2008, lee_first-order_2019} can be employed to design adaptive quantum algorithms \cite{ADAPTVQE,PhysRevA.107.062421, gluza_double-bracket_2025, mcmahon_equating_2025, suzuki_grovers_2025,  PhysRevResearch.5.033227, malvetti_randomized_2024} that  prepare the ground state for almost all initial states \cite{PhysRevResearch.5.033227,malvetti_randomized_2024,gluza_double-bracket_2025, mcmahon_equating_2025, suzuki_grovers_2025}. However, typically a fixed unitary transformation $U = U(\vec{\theta})$ that is parameterized by some variables $\vec{\theta}\in \mathbb R^{M}$ (e.g., rotational angles etc.) is used as a quantum circuit ansatz to solve the optimization problem. In this case, criteria for when a classical optimizer that iteratively updates $\vec{\theta}$ converges to a ground state are not known as parameterizing $U$ can introduce local optima.
This setting is similar to the setting of quantum control where control field parameters are used to steer the time evolution of a quantum system in a desired fashion.

In this sense, the convergence of VQAs is closely related to
finding quantum optimal controls in the form of shaped control fields. In fact, VQAs can be regarded as solving quantum optimal control problems at the circuit level \cite{PRXQuantum.2.010101}.
The ease in finding optimal controls that, e.g., prepare a desired state or implement a target quantum logic gate, has inspired decades of research centered around analyzing quantum control landscapes \cite{Chakrabarti01102007}. A seminal paper \cite{doi:10.1126/science.1093649}  argued that under certain conditions finding optimal controls through numerical optimizers is straightforward, i.e., the numerical optimizer does not get stuck in suboptimal solutions but finds the optimal controls for the desired task. 
More specifically according to \cite{Chakrabarti01102007, russell_control_2017}, if the control system is (i) controllable, (ii) control fields are unconstrained, and (iii) the system can locally be steered in all directions (local surjectivity, as depicted in Fig. \ref{fig:Intro}) by varying the controls, then the quantum control landscape is free of local optima. However, these conditions have sparked debate \cite{PhysRevLett.106.120402,PhysRevLett.108.198901} over whether (iii) can be satisfied, given that systems exist for which (iii) does not hold \cite{doi:10.1142/S0219025713500215,PhysRevA.86.013405}, i.e. that admit singular controls.

While for a single qubit (i.e., a d = 2 dimensional quantum system) there are control systems for which the existence of local optima can be ruled out [32], in general, quantum control systems which are locally surjective when restricting the control functions to a finite dimensional parameter space are, to the best of our knowledge, not known.
Moreover, the absence of local optima does not imply that an optimizer will converge to a global opti-
\begin{figure}[H]
    \centering    \includegraphics[width=\linewidth]{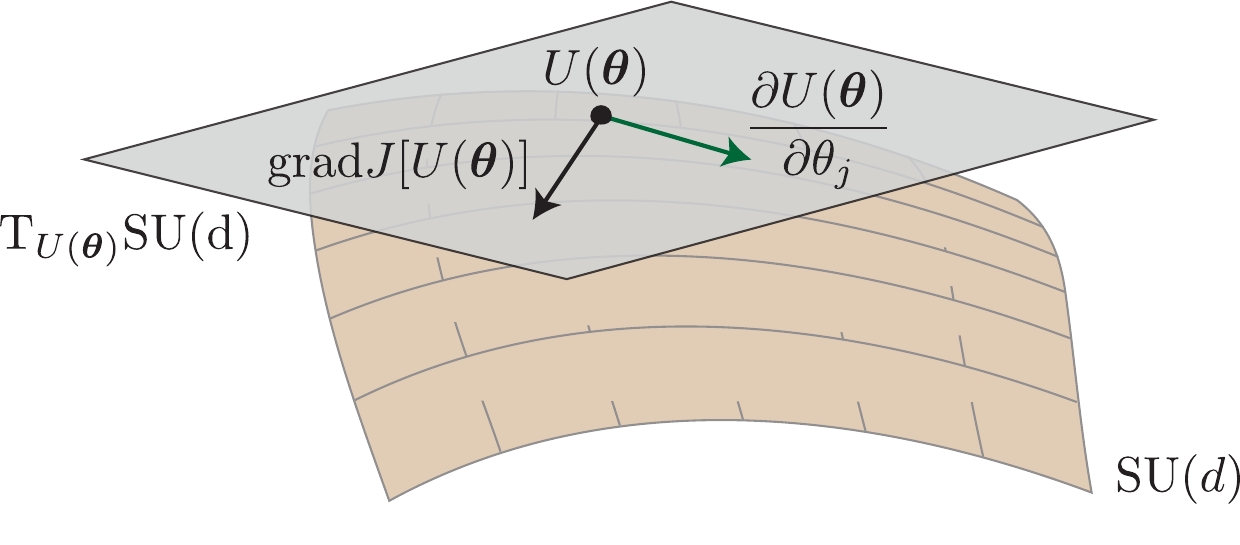}
    \caption{Schematic representation of the key criterion (local surjectivity) that establishes convergence of VQEs to the ground state. If local surjectivity is satisfied, the variation $\frac{\partial U(\bm{\theta})}{\partial\theta_{j}}$ (green) of a parameterized unitary transformation $U(\bm{\theta})\in\mathrm{SU}(d)$ with respect to all variational parameters $\theta_{j}$ spans the tangent space $\text{T}_{U(\bm{\theta})}\mathrm{SU}(d)$ of the special unitary group \(\mathrm{SU}(d)\) at all points of the optimization landscape. In this case the critical point structure of the VQE is determined by the Riemannian gradient $\operatorname{grad}J$. We show that then the Riemannian gradient only vanishes at global optima and strict saddle points that are avoided almost surely by gradient descent when gradients are Lipschitz continuous. 
    Thus, if the gradient descent used for the parameter update converges, it converges for almost all initial configurations to a global minimum (main Theorem).   
    }
    \label{fig:Intro}
\end{figure}
\noindent mum, as saddle points can also hinder the search. 

In this work we develop crucial steps towards a theory of convergence for VQEs and quantum optimal control. This is achieved by providing a systematic way to construct parameterized unitary transformations $U(\bm{\theta})$ that are free of singular controls, hereafter referred to as \emph{singular points}, and thereby satisfy local surjectivity (iii). 
We show that in this case, suboptimal solutions of VQEs correspond to strict saddle points. These saddle points are avoided almost surely by gradient descent algorithms, utilizing a result from non-linear stability theory
\cite{lee_first-order_2019}. 
However, the developed ansatz families satisfying (iii) with unconstrained parameters $\bm{\theta}\in\mathbb R^{M}$ result in a technical hurdle. The convergence of the gradient descent algorithm that updates $\bm{\theta}$ cannot be guaranteed anymore, i.e., the situation that the algorithm escapes to infinity can technically not be ruled out, although numerical evidence is favorable. These convergence results are summarized in Theorem \ref{thm:su(d)_convergence}, which provides a criterion for when VQEs converge to the ground state.

Next we investigate how much of this guarantee survives if the VQE ansatz is restricted to a smaller subgroup of unitaries.
In that case, under similar assumptions as before, we can still guarantee that the VQE converges to a globally (on that subgroup) optimal solution and from representation theory we can derive sufficient conditions on when this global optimum corresponds to a ground state of the Hamiltonian.

We go on to analyze for \emph{the \(\mathbb{SU}(d)\)-gate ansatz} \cite{wiersema_here_2024, wiersema_geometric_2025, wierichs_symmetric_2023, banchi_measuring_2021, chen_special-unitary_2025} and \emph{the product-of-exponentials ansatz}
\cite{cerezo_variational_2021, leone_practical_2024, xiao_physics-constrained_2024, kandala_hardware-efficient_2017, wang_quantum_2023, hadfield_quantum_2019, golden_quantum_2023, wecker_progress_2015, park_hamiltonian_2024, anselme_martin_simulating_2022} whether the assumptions of the Theorem are satisfied. We show that for both families with $M \leq d^{2}-1$ variational parameters singular points where local surjectivity breaks always exist. We go on to establish a stronger result in \Cref{sec:overparameterization} for the \(\mathbb{SU}(d)\)-gate ansatz by proving that singular points cannot be removed regardless of how much the quantum circuit is overparameterized. In light of these challenges, we provide parameterized quantum circuit constructions with $M=2(d^{2}-1)$ and $M=d^{2}$ parameters that satisfy local surjectivity and discuss control field parameterizations for applications in quantum control. We conclude with a discussion about leveraging symmetries to reduce the number of optimization parameters while maintaining the favorable convergence properties and the termination of gradient-descent type algorithms. We briefly address regularization techniques that may avoid gradient descent running to infinity and draw connections to the halting problem.

\section{Results}

\subsection{Setting}

\Glspl{VQA} implement a parameterized unitary transformation \(U(\vec{\theta})\) on a quantum computer.
This unitary is used to estimate a cost function \(J(\vec{\theta})\), which is minimized by a classical optimization routine like gradient descent that iteratively updates the parameters according to the update rule
\begin{equation}
    \label{eq:gradient_descent}
    \vec{\theta}_{k+1} = \vec{\theta}_k - \gamma \nabla J(\vec{\theta}_k).
\end{equation}
Here, $k=0,1,\cdots$ labels the iterative step, \(\gamma\) is the step size of the algorithm, and $\nabla J(\vec{\theta}_k)$ denotes the gradient with respect to $\vec{\theta}$ at the value $\vec{\theta}_k$.

The \gls{VQE} seeks to find the ground state energy of a Hamiltonian \(H\) by minimizing the energy expectation value 
\begin{equation}    
\label{eq:VQE_cost}
J(\vec{\theta}) = \bra{\psi_0} U^\dagger(\vec{\theta}) H U(\vec{\theta}) \ket{\psi_0},
\end{equation}
for some initial state \(\ket{\psi_0}\) that is evolved by \(U(\vec{\theta})\) to prepare the state $\ket{\psi(\vec{\theta})}=U(\vec{\theta})\ket{\psi_{0}}$. Since this cost function is independent of the global phase of $U(\vec{\theta})$, it usually suffices to consider unitary transformations $U(\vec{\theta})$ that live in the special unitary group $\mathrm{SU}(d)$ of unitary $d\times d$ matrices with unit determinant.  

VQEs have a broad range of applications, including computing the ground state energy of molecules in quantum chemistry \cite{VQEChem} and solving combinatorial optimization problems \cite{farhi_quantum_2014}. When $H=\mathds{1}-\ket{V}\bra{V}$ is given by a projector formed by some target state $\ket{V}$, the cost function in \eqref{eq:VQE_cost} becomes the fidelity error $J(\vec{\theta})=1-|\langle V|\psi(\vec{\theta}) \rangle |^{2}$ whose minimization corresponds to a state preparation problem. This is a typical scenario in quantum control where the unitary transformation $U(\vec{\theta})$ correspond to the time evolution operator that solves the Schr\"{o}dinger equation governed by a time dependent Hamiltonian of the form 
$ H_f(t) = H_0 + \sum_{j=1} f_j(t) H_j$ \cite{Brif_2010,
doi:10.1126/science.288.5467.824,Koch}, cf. \Cref{sec:QC_vs_QVAs} for more details.
In this setting, the parameters $\vec{\theta}$ correspond to some parameterization of the control fields $f_{j}(t)$, such as the amplitudes of piecewise constant controls \cite{KHANEJA2005296} or frequencies and phases of the control fields that are expanded in a Fourier series, cf. \cite{Müller_2022} and \Cref{sec:QC_vs_QVAs}. Beyond state preparation, the implementation of a quantum logic gate $V$ in quantum computing can also be achieved by minimizing the gate error $J(\vec{\theta})=1-|\langle V,U(\vec{\theta})\rangle|^{2}$ where here $\langle A,B\rangle=\text{Tr}\{A^{\dagger}B\}$ denotes the Hilbert Schmidt inner product. 

However, despite the broad utility of minimizing cost functions of the form \eqref{eq:VQE_cost} through gradient-descent \eqref{eq:gradient_descent}, criteria for when a parameterized unitary transformation $U(\vec{\theta})$ yields convergence of the corresponding \gls{VQE} to the target ground state are not known. 
The difficulty mainly arises from the parameterization of the unitary group, which might introduce additional critical points and from its unbounded domain.

To make things more precise, we call a smooth map from \(\mathds{R}^M\) into the group \(\mathrm{U}(d)\) of unitary matrices a \emph{parameterization} of the unitary matrices. We want to briefly remark that, contrary to the standard terminology of differential geometry, we do not require any injectivity of the parameterization, as we want to explicitly allow for overparameterization.
Moreover, we always require the parameterization to be defined on the entirety of \(\mathds{R}^M\) rather than just some subset, to avoid situations where the gradient descent leaves the domain of the parameterization.
Further, we want to call a parameterization \emph{surjective}, if it can ``reach'' any unitary up to a global phase factor.
This includes in particular parameterizations, that are surjective on \(\mathfrak{su}(d)\) in the classical sense.

Notice that the cost function \eqref{eq:VQE_cost} can be regarded as a composition of two parts: 1) a parameterization \(U(\vec{\theta})\) of the unitaries and 2) a cost function \(J[U] = \bra{\psi_0} U^\dagger H U \ket{\psi_0}\) defined on the unitary group itself.
By abuse of notation, we call the composition \(J[U(\vec{\theta})]\) of these two parts just \(J(\vec{\theta})\).

As a direct application of the chain rule, the gradient of the cost function with respect to the $j$-th component $\theta_{j}$ of $\vec{\theta}$ can be expressed in terms of the Hilbert-Schmidt inner product of the derivatives of 
$\vec{\theta} \mapsto U(\vec{\theta})$ and $U \mapsto J[U]$, cf.~e.g.\cite{Bro88+91,smith_optimization_1994,SGDH08,PhysRevA.86.013405,10.1063/1.2198837}
as well as \cite{HM94}[Chap.~2.1].
\begin{equation}
    \begin{split}
    \label{eq:gradient_decomposition}
    \frac{\partial}{\partial \theta_j}J(\vec{\theta}) &= \langle \operatorname{grad}J[U(\vec{\theta})], \Omega_j(\vec{\theta}) \rangle\\
    \end{split}
\end{equation}
Here, \(\operatorname{grad}J[U] = [H, U\ket{\psi_0}\bra{\psi_0}U^\dagger]\) is the \textit{Riemannian gradient} with respect to the standard metric on \(\mathrm{U}(d)\) of the cost function $U \mapsto J[U]$, 
i.e.~the direction tangential to the manifold of steepest increase of \(J\) and
\begin{align}
\label{eq:LieAlgElements}
\Omega_j(\vec{\theta}) = U^\dagger(\vec{\theta}) \frac{\partial}{\partial \theta_j} U(\vec{\theta}).
\end{align}
Note that both quantities are pulled back to the tangent space of the unitary group at the identity given by the Lie algebra $\mathfrak{su}(d)$ of skew Hermitian and traceless $d\times d$ matrices. The ``true'' gradient can be obtained by multiplying 
$\operatorname{grad}J[U]$ from the left with $U$. As depicted in Fig. \ref{fig:Intro}, we use this differential geometric perspective to characterize when the gradient descent \eqref{eq:gradient_descent} converges to the global minimum of the cost function \eqref{eq:VQE_cost} that corresponds to the desired target state. 

\subsection{Avoiding suboptimal solutions almost surely}
For analytic cost functions and a sufficiently small steps size $\gamma$ it is well known that gradient descent always either diverges to infinity or converges to a critical point \(\vec{\theta}^*\) at which the gradient vanishes \(\nabla J(\vec{\theta}^*) = 0\) (see Theorem 4.1 in \cite{absil_convergence_2005}). It is also known that the Riemannian gradient of the cost functionals corresponding to equation \eqref{eq:VQE_cost} only vanishes at the, possibly degenerate, global maxima and saddle points (see e.g.~\cite{HM94}[Chap.~1.3, Thm.3.4/Proof 3.5 and Chap.~2.1] and \cite{Duistermaat83})
that correspond to the eigenstates of $H$ (see Methods, \Cref{sec:Methods}). 

Now, if the $\Omega_{j}$'s span the full tangent space for all $\vec{\theta}$, which we refer to as \emph{local surjectivity}, the critical point structure is determined by the Riemannian gradient as the situation that $\operatorname{grad}J[U]$ is nonzero and orthogonal to all $\frac{\partial}{\partial \theta_{j}}U(\vec{\theta})$ cannot occur.
Note that since global phases do not carry physical meaning, usually the cost function does not depend on the overall global phase.
Hence it is sufficient if the $\Omega_{j}$'s span the Lie algebra \(\mathfrak{su}(d)\) of traceless, skew-Hermitian \(d \times d\) matrices.
\begin{definition}[Local surjectivity]
\label{def:local_surjectivity}
We say that a unitary parameterization $\vec{\theta} \to U(\vec{\theta})$ is locally surjective at $\vec{\theta}_0$ if
\begin{equation}
    \mathfrak{su}(d) \subseteq \mathrm{span}\{\Omega_j(\vec{\theta}_0) \mid j=1,...,M\}\,. 
\end{equation}
The parametrization is called \emph{locally surjective} 
if it is \emph{locally surjective} for all $\vec{\theta} \in \mathds{R}^M$.
\end{definition}
\noindent
We note that locally surjective maps are known as \emph{submersions} in differential geometry \cite{lee_introduction_2013}. 
Also, we want to remark that the term local surjectivity at $\vec{\theta}_0$ is sometimes defined to mean a map, which takes neighborhoods of $\vec{\theta}_0$ to neighborhoods of $U(\vec{\theta}_0)$. 
Our notion of local surjectivity implies the latter, but not the 
other way around (Thm.~2.5.9 in \cite{AMR88} or \cite{grasse_higher-order_1986}).

The local surjectivity property ensures that at each point $\vec{\theta}$ in the Euclidean parameter space $U(\vec{\theta})$ has access to any direction on the manifold of the special unitary group, i.e., by varying $\vec{\theta}$ the system can be steered in all tangent space directions.
We call points in the optimization landscape where local surjectivity breaks down \emph{singular points}. Local surjectivity is independent of the problem structure defined by $H$ but only depends on the parameterized unitary transformation $U(\vec{\theta})$. However, there has been a longstanding debate in the quantum control community whether systems exist for which local surjectivity can be satisfied. The more general problem of the existence of submersions from an open manifold to another manifold has also been discussed in the differential geometry literature \cite{phillips_submersions_1967}. 

Below we solve this issue by explicitly constructing unitary transformations that exhibit the local surjectivity property. With further details found in the methods section, we show that then the critical points $\vec{\theta}^{*}$ of $J(\vec{\theta})$ correspond to global optima or saddle points whose Hessian has at least one negative eigenvalue. Such saddle points are known in the classical optimization community as strict saddles points, which are avoided for almost all (e.g., with respect to the Lebesgue measure) initial values by gradient descent provided the gradient is globally Lipschitz continuous \cite{lee_first-order_2019, pmlr-v49-lee16}. Guided by this, we establish the following Theorem.  

\begin{theorem}[Convergence theorem on \(\mathrm{SU}(d)\)]
    \label{thm:su(d)_convergence}
    Let \(H\) be a Hamiltonian and $U(\vec{\theta})\in\mathrm{SU}(d)$ be a real analytic parameterized unitary transformation used as a VQE ansatz to minimize the expectation value \eqref{eq:VQE_cost} of $H$. If local surjectivity holds for $U(\vec{\theta})$ and the Lie algebra elements $\Omega_j(\vec{\theta})$ in \eqref{eq:LieAlgElements} are uniformly bounded, then for a suffiently small stepsize $\gamma$ and almost all initial points \(\vec{\theta}_0\), a VQE updated by the gradient descent algorithm \eqref{eq:gradient_descent} either diverges to infinity, or converges to a ground state of $H$.
\end{theorem}

We remark that the uniform boundedness of the \(\Omega_j(\vec{\theta})\) in an arbitrary matrix norm is a technical requirement that ensures the gradients to be globally Lipschitz continuous (cf. \Cref{sec:lipschitz}).
Further, we expect that rather than analyticity, smoothness of the parameterization already is sufficient (see \Cref{sec:morse-bott} for details).

The Theorem gives guarantees on the convergence of the \gls{VQE}. In particular, it says that, given the algorithm does not run off to infinity but instead converges to some finite critical point \(\vec{\theta}^*\), then, except for a negligible set of initial values, that point does indeed correspond to a global optimum of the cost function. In other words, the algorithm almost never converges to a suboptimal solution.

In the ``few'' extraordinary cases, the algorithm converges to an excited eigenstate of \(H\) instead.
However, the theorem guarantees that these cases are rare in the sense that when sampling the initial point from a reasonable distribution, e.g. Gaussian or uniformly from some finite region, the probability of hitting such an initial point is zero.

\subsection{Optimization on smaller subgroups}
Up until now, we assumed that the optimization is carried out over the entire special unitary group.
Due to the exponential scaling of the dimension of \(d = 2^n\) for a \(n\)-qubit system, parameterizing the entire special unitary group quickly becomes intractable in practical settings.
Therefore it is natural to also study cases where the ansatz is restricted to a much smaller Lie subgroup $G \subseteq \mathrm{SU}(d)$ with Lie algebra $\mathfrak g$. 
In order to avoid very pathological cases like irrational windings of a torus, where local optima surely become unavoidable, we want to assume that the subgroup is closed, or equivalently, compact.
In this section we present conditions on $G$, $H$, and $\lvert\psi_0\rangle$ under which the convergence guarantees of the main theorem survive this restriction.

First, we need a modified definition of local surjectivity that incorporates the fact that we are only optimizing over the subgroup \(G\).
\begin{definition}[Local surjectivity for parameterizations of subgroups]
    We say that a parameterization \(\vec{\theta} \rightarrow U(\vec{\theta}) \in G\) of a closed Lie subgroup \(G \subset \mathrm{SU}(d)\) with Lie algebra \(\mathfrak{g}\) is locally \(\mathfrak{g}\)-surjective at \(\vec{\theta}_0\) if
    \begin{equation}
        \mathfrak{g} = \operatorname{span}\{\Omega_j(\vec{\theta}_0) \lvert j = 1,...,M\}.
    \end{equation}
    The parameterization is called locally \(\mathfrak{g}\)-surjective if it is locally \(\mathfrak{g}\)-surjective at all points \(\vec{\theta} \in \mathds{R}^M\).
\end{definition}

Let us call \(\rho_0 = \ket{\psi_0}\bra{\psi_0}\) and \(P_0 = i(\rho_0 - \frac{1}{d}\mathds{1})\) its traceless, skew-Hermitian part.
The cost function \eqref{eq:VQE_cost} can be equivalently written as
\begin{equation}
    \label{eq:cost_moment_map}
    J(\vec{\theta}) = -\braket{iH, U(\vec{\theta}) P_0 U^\dagger(\vec{\theta})} + \frac{1}{d}\operatorname{Tr}\{H\},
\end{equation}
where \(\langle A, B\rangle = \operatorname{Tr}\{A^\dagger B\}\) denotes the Hilbert--Schmidt inner product and the additional trace term is just a constant that can be omitted.

Now in this form of the cost function closely resembles a so-called \textit{moment map}, which is studied in the field of symplectic geometry.
For an introduction to moment maps, see e.g. Chapter II of \cite{guillemin_symplectic_2001}.

By making use of the well-studied structure of these moment maps, we are able to carry over the essential parts of the proof of the main theorem and establish the following results.

\begin{theorem}[Convergence theorem on subgroups]
    \label{thm:subgroup_convergence}
    Let \(H\) be a Hamiltonian, \(G \subset \mathrm{SU}(d)\) be a closed Lie subgroup and \(\mathfrak{g} \subset \mathfrak{su}(d)\) be its compact Lie algebra.
    Let $U(\vec{\theta})\in G$ be a real analytic parameterized unitary transformation used as a VQE ansatz to minimize the expectation value \eqref{eq:VQE_cost} of $H$.
    Assuming that local \(\mathfrak{g}\)-surjectivity holds for $U(\vec{\theta})$, the Lie algebra elements $\Omega_j(\vec{\theta})$ in \eqref{eq:LieAlgElements} are uniformly bounded and that \(iH \in \mathfrak{g}\), then for a suffiently small stepsize $\gamma$ and almost all initial points \(\vec{\theta}_0\), a VQE updated by the gradient descent algorithm \eqref{eq:gradient_descent} either diverges to infinity, or converges to the global optimum of the cost on \(G\).
    This global optimum of the cost always corresponds to an eigenstate of \(H\) if \(P_0 = i(\ket{\psi_0}\bra{\psi_0}- \frac{1}{d}\mathds{1})\) is contained in some Lie algebra \(\mathfrak{h}\) such that \(\mathfrak{g}\) forms an ideal in \(\mathfrak{h}\), i.e. \([\mathfrak{g}, \mathfrak{h}] \subseteq \mathfrak{g}\).
\end{theorem}

The full proof of every statement can be found in Appendix \ref{app:subgroup-convergence}.
We want to remark that in principle one can always replace \(iH \leftrightarrow P_0\) and \(U(\vec{\theta}) \leftrightarrow U^\dagger(\vec{\theta})\) in \eqref{eq:cost_moment_map}, to obtain an equivalent cost function which only differs by an irrelevant additive constant.
This effectively allows one to swap the assumptions on \(iH\) and \(P_0\) at will in the Theorem.
Also, notice that the special case in which both \(iH, P_0 \in \mathfrak{g}\) satisfies the extra assumption in the Theorem that allows one to conclude that the global optimum is indeed an eigenstate of \(H\).
Recently, a very similar result to Theorem \ref{thm:subgroup_convergence} was independently discovered in \cite{chen_quantum_2026}.

Moreover, assume that a ground state of \(H\) is reachable by some \(U(\vec{\theta}) \in G\) from \(\ket{\psi_0}\).
In that case the global optimum of the cost on \(G\) must correspond to a ground state of \(H\) and we effectively recover the same guarantees as in Theorem \ref{thm:su(d)_convergence}.
That begs the question of whether there are any easily checkable, sufficient conditions under which we can guarantee that some ground state of \(H\) is indeed reachable with \(G\) from some initial state \(\ket{\psi_0}\).
We find that this is indeed possible with the use of the representation theory of compact Lie groups and their algebras.

A representation is a map that takes the abstract Lie algebra \(\mathfrak{g}\) and maps it to specific matrices acting on some Vector space.
In that sense choosing a specific parameterization \(U(\vec{\theta})\) of \(G\) also fixes its representation.
There is a special set of representations, which are called quasi-minuscule.
A complete definition of quasi-minuscule representation is given in Appendix \(\ref{app:subgroup-convergence}\).
For now it is only important that we can guarantee for any Hamiltonian \(iH \in \mathfrak{g}\) that a ground state of \(H\) is reachable with \(G\) from a known set of initial states (the so-called weight vectors of the representation, see Appendix \ref{app:subgroup-convergence}) if the representation is quasi-minuscule and that the quasi-minuscule representations of semi-simple Lie algebras are completely categorized \cite{lee_polytopes_2017}.
Therefore applying this criterion only amounts to identifying the representation and comparing it against a complete list of examples.
We summarize this in the following Theorem.
\begin{theorem}
    \label{thm:quasi-minuscule}
    If the compact Lie group \(G\) with Lie algebra \(\mathfrak{g}\) is represented on the Hilbert space \(\mathcal{H}\) by a quasi-minuscule representation and \(H\) is a Hamiltonian with \(iH \in \mathfrak{g}\), then a ground state of \(H\) is reachable with an element of \(G\) from any weight vector \(\ket{\psi_0}\) of the representation.
\end{theorem}

\subsection{Application to free fermions}
\label{sec:free_fermions}
Here, we will apply the results of the previous section to the task of finding the ground state of a free fermionic Hamiltionian.
We briefly recount here how free fermions can be simulated on a quantum computer (see e.g. \cite{chien_simulating_2026}).
Consider a fermionic quantum system with \(n\) modes.
Due to the Pauli exclusion principle, any mode can only be either occupied or unoccupied.
The Hilbert space describing this system is therefore \(\mathds{C}^{2^n}\), i.e. the same as for \(n\) qubits.
The Majorana operators \(\gamma_j = a_j + a^\dagger_j\) and \(\bar{\gamma}_j = -i (a_i - a^\dagger_i)\), where \(a_j\) and \(a^\dagger_j\) are the fermionic creation and annihilation operators, can be represented on \(\mathds{C}^{2^n}\) using the Jordan-Wigner transformation
\begin{equation}
    \gamma_j \mapsto (\prod_{k<j} \sigma_z^k) \sigma_x^j, \qquad \bar{\gamma_j} \mapsto (\prod_{k<j} \sigma_z^k) \sigma_y^j,
\end{equation}
where the upper index on the Pauli matrices denotes which qubit it acts on.

Consider the Lie algebra of all Hamiltonians, which are quadratic in the Majorana operators.
Under the Jordan-Wigner transformation, this Lie algebra is given by \cite{goh_lie-algebraic_2025}
\begin{align}
    \mathfrak{g} = & \left\langle\left(\bigcup_{\mu, \nu = x, y} \{i\sigma_\mu^j \sigma_\nu^{j+1} \mid j = 1,...n-1\}\right) \right. \nonumber \\
     & \hphantom{\Bigg\langle\Bigg(} \cup \left. \{i \sigma_z^j \mid j=1,...,n\} \vphantom{\bigcup_{\mu, \nu = x, y}} \right\rangle_\text{Lie},
\end{align}
which is isomorphic to \(\mathfrak{so}(2n)\).
Its dimension, and hence also the number of parameters necessary for a locally surjective parameterization, scales only polynomially in \(n\), namely \(\operatorname{dim} \mathfrak{g} = n(2n-1)\).

The representation is reducible and splits into a direct sum of irreducible, minuscule representations on each sector of the parity operator \(P = \prod_{j=1}^n \sigma_z^j\).
The corresponding highest-weight vectors are $\ket{\Omega_+} = \ket{0\,0\cdots 0}$ and $\ket{\Omega_-} = \ket{0\,0\cdots 1}$.
Since all elements of \(\mathfrak{g}\) commute with the parity operator, they have ground states of definite parity.
Hence Theorem \ref{thm:quasi-minuscule} applies on each parity sector separately and with Theorem \ref{thm:subgroup_convergence} we get that for almost all initial points that a \gls{VQE} with a real analytic, locally \(\mathfrak{so}(2n)\)-surjective parameterization of \(\mathrm{SO}(2n)\) and uniformly bounded Lie algebra elements either diverges to infinity or converges to a ground state of any free fermionic Hamiltonian \(H\), starting from either \(\ket{\psi_0} = \ket{\Omega_+}\) or \(\ket{\psi_0} = \ket{\Omega_-}\).

\subsection{Gimbal locks in commonly used parameterized quantum circuits}

To gain some intuition for the importance of local surjectivity, we start by considering a single qubit example. 

The Hilbert space of a single qubit has complex dimension \(d=2\) and the Lie algebra of traceless skew Hermitian \(2 \times 2\) matrices \(\mathfrak{su}(2)\) has real dimension three.
A basis for \(\mathfrak{su}(2)\) is, up to a prefactor of \(i\), given by the Pauli matrices
\begin{equation}
    \sigma_X = \begin{bmatrix}
        0 & 1 \\ 1 & 0
    \end{bmatrix}, \  \sigma_Y = \begin{bmatrix}
        0 & -i \\ i & 0
    \end{bmatrix}, \  \sigma_Z = \begin{bmatrix}
        1 & 0 \\ 0 & -1
    \end{bmatrix}.
\end{equation}
A common parameterization for \(\mathrm{SU}(2)\) are the Euler angles, defined in the \(X-Y-X\) convention as
\begin{equation}
    \label{eq:euler_angles}
    U(\vec{\theta}) = \exp(-i\sigma_X\theta_3) \exp(-i\sigma_Y\theta_2) \exp(-i\sigma_X\theta_1).
\end{equation}
Indeed, any single qubit unitary can be written in this form up to an irrelevant global phase. However, despite the existence of parameters $\vec{\theta}=(\theta_{1},\theta_{2},\theta_{2})$ that achieve any unitary (reachability), local optimization methods may not be able to find the parameters that achieve a desired unitary transformation, which we discuss below. 

If, for example, the Riemannian gradient points into the \(\sigma_Z\)-direction at the identity, then a gradient algorithm initialized at the origin \(\vec{\theta} = (0, 0, 0)^\mathrm{T}\) will immediately get stuck, as only the \(\sigma_X\) and \(\sigma_Y\) directions are available.
Note that even though \(U(\vec{\theta})\) may not be able to change into all directions at certain points, \(\ket{\psi(\vec{\theta})}\) might still be.
However, we can also construct situations, where the problem arises both on the level of \(U(\vec{\theta})\) and \(\ket{\psi(\vec{\theta})}\), as the next example shows.
Assume that the gradient descent update \eqref{eq:gradient_descent} is at the point \(\vec{\theta} = \left(0, \frac{\pi}{2}, 0\right)^\mathrm{T}\) at which the Lie algebra elements $\Omega_{j}(\vec{\theta})$ are given by \(\Omega_1(\vec{\theta}) = -i\sigma_X\), \(\Omega_2(\vec{\theta}) = -i \sigma_Y\) and \(\Omega_3(\vec{\theta}) = i\sigma_X\).
Consequently, if the Riemannian gradient points into the \(\sigma_Z\)-direction at this point, then the Lie algebra elements have
no suppport in the direction of the Riemannian gradient and the Euclidean gradient will vanish. As such, the optimization routine
\begin{figure}[H]
    \centering
    \includegraphics[width=\linewidth]{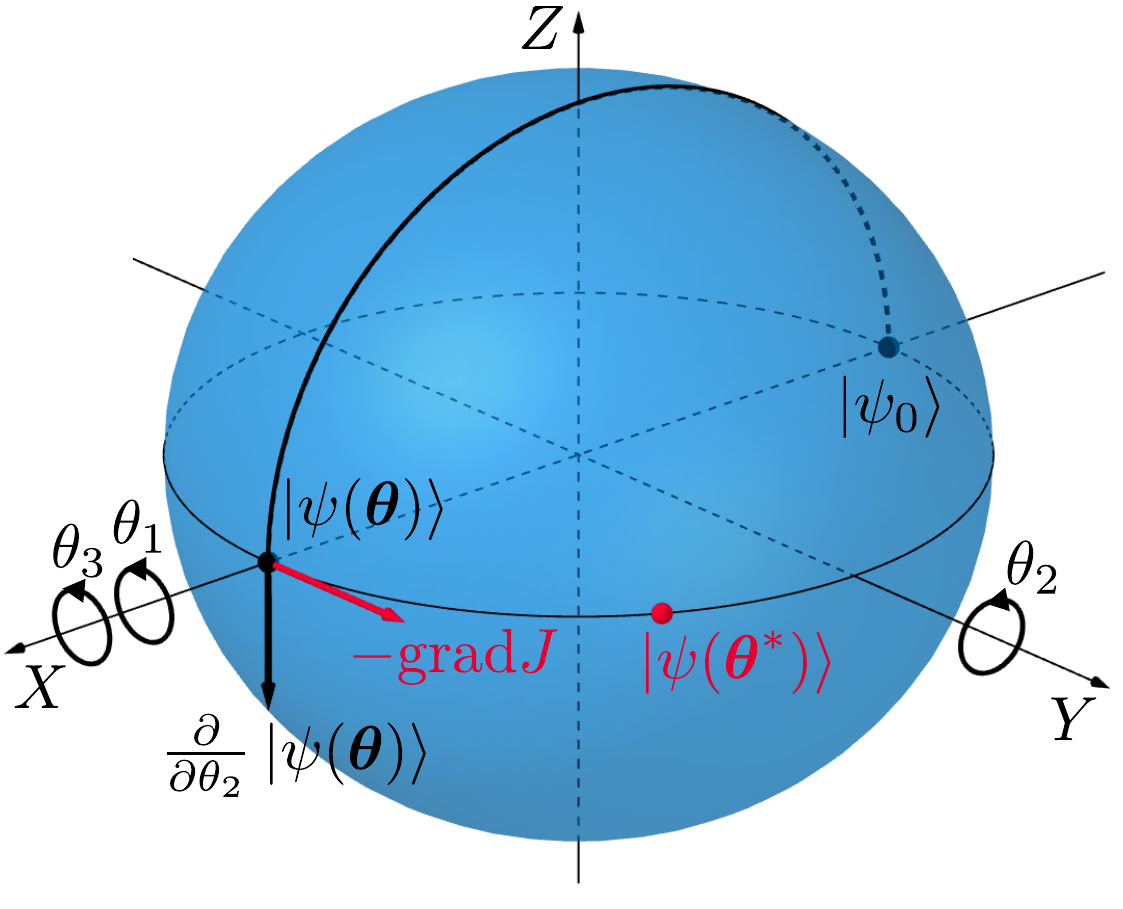}
    \caption{Example of a singular point in the Euler angle parameterization of \(\mathrm{SU}(2)\) on the Bloch sphere.
    We consider a \gls{VQE} that uses the X-Y-X Euler angle parameterization of \(\mathrm{SU}(2)\) given by equation \eqref{eq:euler_angles} to find the ground state \(\ket{\psi(\vec{\theta}^*)}\) (red dot) of some problem Hamiltonian, e.g. \(H = \mathds{1} - \ket{\psi(\vec{\theta^*})}\bra{\psi(\vec{\theta^*})}\). The first and last rotation will rotate the state around the \(X-\)axis on the Bloch sphere, the second rotation around the \(Y\)-axis. In principle one can map any point on the Bloch sphere to any other point in this way.
    However, when choosing \(\ket{\psi_0} = \ket{-}\), i.e. the minus one eigenstate of \(\sigma_x\), as the initial state, at any point \(\vec{\theta} = {\left(\theta_1, \frac{\pi}{2}, \theta_3\right)}^T\) of the parameter landscape, only the second rotation will affect the state and move it to the opposite end of the Bloch sphere.
    Hence the derivatives \(\frac{\partial}{\partial \theta_1} \ket{\psi(\vec{\theta})}\) and \(\frac{\partial}{\partial \theta_3} \ket{\psi(\vec{\theta})}\) vanish and \(\frac{\partial}{\partial \theta_2} \ket{\psi(\vec{\theta})}\) points along the \(Z\)-direction.
    Since the cost functional only decreases further when the state is moved along the equator, the Riemannian gradient (in this case, defined on the Bloch sphere itself) points into the \(Y\)-direction and is therefore perpendicular to all available derivatives.
    Hence, the Euclidean gradient vanishes and the optimization stops at this point --- although the ground state has not yet been reached.
    }
    \label{fig:bloch-gimbal-lock}
\end{figure}
\noindent terminates at \(\vec{\theta}\), even though the optimal solution may not have been reached yet.

The same issue arises on the level of the state \(\ket{\psi(\vec{\theta})}\) instead of the unitary \(U(\vec{\theta})\), which is illustrated in \Cref{fig:bloch-gimbal-lock}.
The effective loss of a degree of freedom at certain points in the parameterization landscape is well known as \textit{gimbal lock} for the Euler-angle parameterization of rotations in a three-dimensional 
space~\cite{BMM13}, as prominently encountered, e.g., in the Apollo-11 Program~\cite{HOR18}.

The existence of such singular points, at which local surjectivity breaks down carries over from single qubit examples to higher dimensions, e.g., for two commonly employed quantum circuit parameterizations, which we introduce next. The quantum circuit defined by 
\begin{align}
\label{def:canonical_coordinates_first_kind}
U(\vec{\theta})=\exp(X(\vec{\theta})),~~~~X(\vec{\theta})=\sum_{j=1}^{d^{2}-1}\theta_{j}X_{j},
\end{align} 
where $\{X_{j} \mid j = 1, ..., d^2-1\}$ is a basis for $\mathfrak{su}(d)$, is referred to as \emph{the \(\mathbb{SU}(d)\)-gate ansatz} \cite{wiersema_here_2024, wiersema_geometric_2025, wierichs_symmetric_2023, banchi_measuring_2021, chen_special-unitary_2025}. In the context of Lie groups, the corresponding parameters are in this case also called \textit{the canonical coordinates of the first kind} \cite{hilgert_structure_2012}. 
Another large class of ansätze in the literature is of the form
\begin{align}
\label{def:canonical_coordinates_second_kind}
U(\vec{\theta})=\prod_{j = 1}^{M} \exp(\theta_j X_j),
\end{align}
for some arbitrary set of generators \(\{X_j \mid j = 1,...,M\}\).
This includes for example the hardware efficient ansatz \cite{cerezo_variational_2021, kandala_hardware-efficient_2017, leone_practical_2024, xiao_physics-constrained_2024}, the quantum alternating operator ansatz \cite{farhi_quantum_2014, wang_quantum_2023, hadfield_quantum_2019, golden_quantum_2023} and the Hamiltonian variational ansatz \cite{wecker_progress_2015, park_hamiltonian_2024, anselme_martin_simulating_2022}. We will refer to this family of ans\"atze as \emph{product-of-exponentials ansätze}.
Since in this case \(\Omega_j (\vec{0}) = X_j\), it immediately follows that a product-of-exponentials ansatz can only be locally surjective if the generators span the full Lie algebra.
This is usually not given for the hardware efficient or quantum alternating operator ansatz, hence they are susceptible to gimbal lock at the origin of the parameter space.
For this reason we only consider product-of-exponentials ansätze where the \(M = d^2-1\) generators form a basis of \(\mathfrak{su}(d)\) from this point onwards. In this case, the parameters are also called \textit{canonical coordinates of the second kind} by the Lie group community \cite{hilgert_structure_2012}.

Since \(\operatorname{dim} \mathfrak{su}(d) = d^2-1\) it is clear that at least \(d^2-1\) parameters are required for local surjectivity.
One can prove that for \(M = d^2 - 1\) parameters the \(\mathbb{SU}(d)\)-gate and product-of-exponentials ansatz still always admit singular points and therefore cannot be locally
surjective (see page 313 in \cite{hilgert_structure_2012} and \cite{altafini_use_2002}).
Hence, they are in principle vulnerable to the gimbal lock problem described above.  The properties of the different ans\"atze are summarized in \Cref{tab:ansätze}.
Note that in both cases the problem arises from the fact that $\vec{\theta}$ is required to range over all of $\mathbb{R}^M$. Restricting $\vec{\theta}$ to suitable open subsets or non-linear rescaling can in principle remove these singular points (cf. \Cref{sec:construction}) at the cost of losing surjectivity.

\renewcommand{\arraystretch}{1.5}
\begin{table*}
    \centering
    \begin{tabular}{c|c|c|c}
    \hline\hline
         Ansatz & Surjective & Locally surjective & \# of parameters \\
         \hline\hline
         \(\mathbb{SU}(d)\)-gate \cite{wiersema_here_2024, wiersema_geometric_2025, wierichs_symmetric_2023, banchi_measuring_2021, chen_special-unitary_2025} & yes & no & \(d^2-1\) \\
         \makecell{Overparameterized \(\mathbb{SU}(d)\)-gate} & yes & no & \(> d^2-1\) \\
         \makecell{product-of-exponentials \cite{cerezo_variational_2021, leone_practical_2024, xiao_physics-constrained_2024, kandala_hardware-efficient_2017, wang_quantum_2023, hadfield_quantum_2019, golden_quantum_2023, wecker_progress_2015, park_hamiltonian_2024, anselme_martin_simulating_2022}} & \(\text{no}^*\) & no & \(\leq d^2-1\) \\
         \makecell{Overparameterized product of \\ exponentials  \cite{you_convergence_2022, wierichs_avoiding_2020, larocca_theory_2023, liu_analytic_2023, wiersema_exploring_2020, lee_progress_2021}} & \makecell{yes, with \\ enough parameters} & unknown & \(> d^2-1\) \\
         Composite (Eq. \eqref{eq:composite_ansätze}) & yes & yes & \(2(d^2-1)\) \\
         Cayley transform & \(\text{yes}^{**}\) & yes & \(d^2\)\\
         \hline
    \end{tabular}
    \caption{Comparison of the different ansätze described in this work and their properties. Surjectivity in our context means that any unitary in \(\mathrm{U}(d)\) can be reached up to a global phase factor, while local surjectivity means that the partial derivatives span the tangent space of \(\mathrm{SU}(d)\) at any point in the parameter landscape (cf. \Cref{def:local_surjectivity}). The \(\mathbb{SU}(d)\)-gate ansatz requires \(d^2-1\) parameters, is surjective on \(\mathrm{SU}(d)\) but is never locally surjective even with the use of overparameterization (see \Cref{sec:overparameterization}). The general product-of-exponentials ansatz (*) need not be surjective, is not locally surjective and typically uses less than \(d^2-1\) parameters.
    With sufficient overparameterization it becomes surjective \cite{dalessandro_uniform_2021}, however it remains unknown if it also becomes locally surjective. The composite ansätze  introduced in this work are both surjective and locally surjective and require \(2(d^2-1)\) parameters. (**) The Cayley transform is technically speaking not surjective onto \(\mathrm{U}(d)\) or even \(\mathrm{SU}(d)\) in the mathematical sense, but it can still reach any unitary up to an irrelevant global phase, and hence fulfills our notion of surjectivity from above. Also, it can be modified to be surjective onto \(\mathrm{SU}(d)\) in the mathematical sense (see \Cref{sec:modified_cayley}). It is locally surjective and uses \(d^2\) parameters. Hence, the requirements for the main theorem hold only for the composite ansätze and for the Cayley transform and we obtain almost guaranteed convergence to the global optimum if the gradient descent terminates.}
    \label{tab:ansätze}
\end{table*}

In the \gls{VQA} community it is common to overparameterize an ansatz. That is, by adding more variational parameters $M>d^{2}-1$ e.g. to the quantum circuits defined in \eqref{def:canonical_coordinates_first_kind} and \eqref{def:canonical_coordinates_second_kind}, the hope is to reduce the number of local minima in the optimization landscape and improve the overall trainability of the model \cite{larocca_theory_2023}.
However, in \Cref{sec:Methods} (Methods) we show that this strategy cannot succeed in completely removing singular points from the quantum circuit in equation \eqref{def:canonical_coordinates_first_kind}.
For the product-of-exponentials ansatz \eqref{def:canonical_coordinates_second_kind} it remains an open problem if singular points can be completely removed through overparameterization, but numerical evidence suggests otherwise (see \Cref{sec:numerical_evidence}).

We remark that the study of these parameterized unitary transformations and their singular points is by far not limited to the field of \glspl{VQA}.
For example, canonical coordinates of the second kind are widely used in the field of robotics \cite{brockett_robotic_1984}, where the singular points correspond to kinematic singularities of the robot, which prove to be a challenge for control algorithms \cite{donelan_kinematic_2010}.

\subsection{Construction of locally surjective parameterizations}
\label{sec:construction}
Now that we have established that even common parameterized quantum circuits are plagued by singular points at which gradient descent can get stuck, we introduce alternative parameterizations which are still surjective but do not suffer from this issue, i.e., which are locally surjective.
We will leverage two key insights to construct locally surjective parameterized unitary transformations: 
\begin{enumerate}
    \item Both the \(\mathbb{SU}(d)\)-gate ansatz and the product-of-exponentials ansatz have their singular points outside a ball with nonzero radius \(R\) in the parameter space, when the \(X_j\) span \(\mathfrak{su}(d)\).
    \item If we divide the parameterization into two parts
    \begin{equation}
        U(\vec{\theta}) = W(\theta_{k+1}, ..., \theta_n) V(\theta_1, ..., \theta_k),
    \end{equation} then \(\Omega_1(\vec{\theta})\) to \(\Omega_k(\vec{\theta})\) do not depend on \(\theta_{k+1}, ...,\theta_n\).
\end{enumerate}
If one naively multiplies two copies of any of the aforementioned ansätze and restricts the parameter domain of one of the copies to a suitably small region around zero, this already eliminates all singular points. The first \(d^2-1\) Lie algebra elements are independent of all other parameters and hence always span \(\mathfrak{su}(d)\).
However, one cannot make sure that the optimization routine does not leave the domain.
Therefore, we artificially blow up this domain, e.g. by taking a suitably scaled arc-tangent.

We further require that the parameterization can in principle reach any unitary up to global phase.
We want to call this property in this context \emph{surjectivity}, or in quantum optimal control contexts it is also called \emph{controllability}. To summarize, we construct a surjective and locally surjective parameterized unitary transformation as follows.

\begin{proposition}[Composite ans\"atze]
    Let \(\vec{\theta}, \vec{\phi} \in \mathds{R}^{d^2-1}\), \(V(\vec{\phi})\) be either the \(\mathbb{SU}(d)\)-gate ansatz or the product-of-exponentials ansatz and \(R > 0\) be such that all singular points of \(V\) are outside the ball of radius \(R\) in parameter space.
    Further, let \(W(\vec{\theta})\) be any surjective parameterization of \(\mathrm{SU}(d)\) with uniformly bounded derivatives. In particular, a generalized Euler angle ansatz \cite{tilma_generalized_2002, dalessandro_optimal_2004} is a suitable choice for \(W(\vec{\theta})\).
    We construct a surjective and locally surjective parameterized unitary transformation of \(\mathrm{SU}(d)\) as
    \begin{equation}
        \label{eq:composite_ansätze}
        U(\vec{\theta}, \vec{\phi}) = W(\vec{\theta}) V \left( \frac{2R}{\pi} \arctan(\vec{\phi}) \right),
    \end{equation}
    where \(\arctan(\vec{\phi})\) is to be understood as element-wise application of the \(\arctan\) function.
    
\end{proposition}
We refer to this family of ansätze as the \emph{composite ansätze}, as they are composed of two different components \(V\) and \(W\).
The order of \(V\) and \(W\) can be interchanged, but we will only prove local surjectivity for this particular ordering.

In the context of quantum optimal control, a composite ansatz can be for example be realized by a control system with access to an orthonormal basis of the Lie algebra as controls.
The pulse should be split into two parts.
The first part should be restricted such that the pulse area remains smaller than \(R\) on each component.
The second part can be any control sequence that corresponds to a controllable system, i.e., for a sufficiently long evolution time and unconstrained control fields it should be able to reach any unitary \cite{jurdjevic_control_1972}.
We are aware that the interactions necessary for each of the exponentials in the composite ansatz (e.g., arbitrary Pauli strings) are typically not natively supported by the hardware.
They could be approximated through Trotterization \cite{dalessandro_general_2009}.
If and how this composite ansatz can be applied to more generic control systems is left for future study.

It remains to find a suitable value for the radius \(R\) such that the singular points of \(V\) lie outside the open ball.
For the \(\mathbb{SU}(d)\)-gate ansatz, an appropriate value is related to the so called \textit{injectivity radius} of the unitary group, which is per definition the largest radius for which the exponential map does not have any singular points inside the ball of that radius, measured in the Frobenius norm around the origin. For the unitary group the injectivity radius is given by \(\sqrt{2} \pi\) (cf. e.g. p. 313 of \cite{hilgert_structure_2012}).
Hence, when the basis \(X_j\) is chosen to be orthonormal with respect to the standard Hilbert-Schmidt inner product, \(R = \frac{\sqrt{2} \pi}{\sqrt{d^2-1}}\) is a suitable choice.
For the product-of-exponentials ansatz, one needs to manually compute the singular points and choose the radius accordingly.

Let us illustrate the composite ansatz construction on the toy example of a single qubit from earlier.
Consider the product-of-exponentials ansatz for \(\mathrm{SU}(2)\)
\begin{equation}
    G(\vec{\phi}) = \exp(-i\sigma_X\phi_3)\exp(-i\sigma_Y\phi_2)\exp(-i\sigma_Z\phi_1).
\end{equation}
The Lie algebra elements \(\Omega_j\) can be calculated to be
\begin{align}
    & \Omega_1= -i\sigma_Z \\
    & \Omega_2 = -i\cos(2\phi_1)\sigma_Y + i \sin(2\phi_1)\sigma_X \\
    & \Omega_3 = -i\cos(2\phi_1)\cos(2\phi_2)\sigma_X \\
    & \hphantom{\Omega_3 = } -i\sin(2\phi_1)\cos(2\phi_2)\sigma_Y + i\sin(2\phi_2)\sigma_Z \nonumber
\end{align}
In order to check whether these Lie algebra elements are linearly independent, we compute the determinant
\begin{align}
    &\begin{vmatrix}
        0 & i\sin(2\phi_1) & -i \cos(2\phi_1)\cos(2\phi_2) \\
        0 & -i\cos(2\phi_1) & -i \sin(2\phi_1)\cos(2\phi_2) \\
        -i & 0 & i \sin(2\phi_2)
    \end{vmatrix} \\
    &= -\cos(2\phi_2) \nonumber.
\end{align}
Thus, as long as \(\phi_2 < \frac{\pi}{4}\), the $\Omega_{j}$'s span \(\mathfrak{su}(2)\) 
We conclude that for a single qubit the parameterized unitary transformation, 
\begin{align}
    U(\vec{\theta}, \vec{\phi}) = &\exp(-i\sigma_X\theta_3)\exp(-i\sigma_Y\theta_2)\exp(-i\sigma_Z\theta_1) \nonumber\\
    & \cdot \exp\left(-\frac{i}{2}\sigma_X\arctan(\phi_3)\right) \nonumber\\
    & \cdot \exp\left(-\frac{i}{2}\sigma_Y\arctan(\phi_2)\right) \\
    & \cdot \exp\left(-\frac{i}{2}\sigma_Z\arctan(\phi_1)\right), \nonumber
\end{align}
is locally surjective. This construction effectively doubles the number of necessary parameters. It is natural to ask whether it is possible to find locally surjective parameterizations with even fewer parameters.
In the literature it is claimed that in particular product-of-exponentials ans\"atze with exactly \(d^2-1\) parameters always lead to singular
points due to the semisimplicity of $SU(d)$ \cite{altafini_use_2002}.
However, no rigorous proofs is known 
to us. Nevertheless, we expect that at least \(d^2\) parameters are required to avoid singular points.
Indeed, one can construct a locally surjective parameterized unitary transformation with exactly \(d^2\) parameters using the \textit{Cayley transform}
\begin{align}
\label{eq:Cayley}
    \mathrm{cay}: \mathfrak{u}(d) &\rightarrow \mathrm{U}^*(d) \\
    X &\mapsto (\mathds{1}-X)^{-1}(\mathds{1} + X),
\end{align}
where we denote by \(\mathrm{U}^*(d)\) the set of unitaries \(G\), s.t \(\det(G + \mathds{1}) \neq 0\) and \(\mathfrak{u}(d)\) is the unitary algebra of all skew Hermitian matrices. It is a diffeomorphism onto its image, implying that there are no singular points on the entire parameter space \cite{quillen_superconnection_1988}. 
Even though its image is a proper subset of the unitary group, the Cayley transform can reach any unitary up to an irrelevant global phase factor. Also, it can be modified such that its image becomes \(\mathrm{SU}(d)\) instead (see \Cref{sec:modified_cayley}). Hence, we indeed have found a locally surjective unitary transformation that only requires \(d^2\) parameters. 

\begin{corollary*}
    Let \(X_j\), \(j=1,...,d^2\) be a basis of \(\mathfrak{u}(d)\) and \(X(\vec{\theta}) = \sum_{j=1}^{d^2}\theta_j X_j\) for \(\vec{\theta}\in \mathds{R}^{d^2}\). Then the parameterized unitary transformation 
    \begin{align}
    \label{eq:CayleyAnsatz}
        U(\vec{\theta})  = \mathrm{cay}(X(\vec{\theta}))
    \end{align}
    is locally surjective where $\text{cay}(\cdot)$ denotes the Cayley transform defined in \eqref{eq:Cayley}. 
\end{corollary*}

We expect that the parameterized quantum circuit \eqref{eq:CayleyAnsatz} that utilizes the Cayley transform can be implemented on quantum computers through quantum signal processing and block encoding strategies \cite{PhysRevLett.118.010501}.  

\section{Discussion}
In this work we assumed that the Riemannian gradient can a priori point in any direction.
This is however not the case, as the cost function only depends on the state \(\ket{\psi(\vec{\theta})} = U(\vec{\theta})\ket{\psi_0}\) rather than the unitary itself.
Hence in principle, at any state \(\ket{\psi}\) we can factor out the stabilizer subgroup, i.e., the group of special unitaries that leave \(\ket{\psi}\) invariant.
For example, we can restrict the exponential map onto the orthogonal complement in \(\mathfrak{su}(d)\) of the commutant of the projector onto the initial state \(\ket{\psi_0}\bra{\psi_0}\).
The resulting manifold only has a dimension of \(2(d-1)\) instead of \(d^2-1\).

Theorem \ref{thm:subgroup_convergence} offers a path to scalable convergence guarantees, as the subgroup in question can be chosen of polynomial size with respect to the number of qubits, like in the free fermion example discussed in section \ref{sec:free_fermions}.
This result seems surprising in the light of \cite{VQA_traps, anschuetz_critical_2021}, which seem to predict the abundance of bad local minima when the parameter count is smaller than the Hilbert space dimension.
The apparent contradiction is resolved by noting that the results in \cite{VQA_traps, anschuetz_critical_2021} hold on average for product-of-exponentials ansätze with generic generators.
But our Theorem \ref{thm:subgroup_convergence} implicitly assumes that the generators only generate the proper Lie subalgebra \(\mathfrak{g} \subset \mathfrak{su}(d)\).
This property is rare in the sense that that any two generic generators will already generate the full \(\mathfrak{su}(d)\).
So, for the average product-of-exponentials ansatz, Theorem \ref{thm:subgroup_convergence} does not hold and the results of \cite{VQA_traps, anschuetz_critical_2021} apply instead.
We want to remark that recent work indicates that problems with a polynomially sized Lie algebra can in some cases be efficiently solved on classical computers as well \cite{goh_lie-algebraic_2025}.

However, for some parameterizations that satisfy local surjectivity, gradient descent may not converge at all. To demonstrate this issue, let \(U(\vec{\theta})\) be the \(\mathbb{SU}(d)\)-gate ansatz and consider for some very small radius \(R\) the parameterization \(U(\frac{2R}{\pi}\arctan(\vec{\theta}))\).
While this parameterization satisfies all the conditions of the theorem, only unitaries that are sufficiently close to the identity can
\begin{figure}[H]
    \centering
    \includegraphics{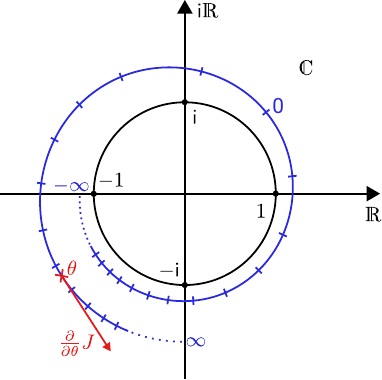}
    \caption{Example of a gradient algorithm escaping to infinity for an inappropriate choice of parameterization.
    The optimization is done on the complex unit circle \(\mathrm{U}(1)\) (black).
    The real line of parameters (blue) is wrapped around the circle such that the negative half extends up to the point \(-1\), but never reaches it. The positive half is wrapped the other way around until \(- i\) so that it overlaps with the negative half in the third quadrant of the circle.
    The increasing radius of the blue spiral only serves visibility purposes.
    When a gradient search with the target \(1\) is initialized at a point \(\theta\) that is mapped to the third quadrant of the circle, the gradient (red) will always point into the positive direction, driving the parameter towards positive infinity, without the image point on the circle ever reaching \(1\).}
    \label{fig:escape_to_infinity}
\end{figure}
\noindent be reached. Consequently, if the global optimum lies outside this region, the gradient algorithm must escape to infinity.

This problem is not only restricted to parameterized unitary transformations that cannot reach the global optimum but can also occurs when the parameterization is surjective.
We start by explicitly illustrating this behavior on the simpler manifold \(\mathrm{U}(1)\), i.e. the unit circle, considering the parameterization $ U(\theta) = e^{i \left( \frac{5}{2} \arctan(\theta) + \frac{\pi}{4}\right)}$ (see Figure \ref{fig:escape_to_infinity} for visualization).
The negative half of the real line is mapped such that it gets arbitrarily close to the point \(-1\), but never reaches it.
The positive half however is wrapped around the circle up to the point \(-i\), so it reaches \(-1\) and overlaps with the part that is also reached by the negative half.
This parameterization fulfills all the conditions of the main theorem.
However, if one constructs a cost function such that \(1\) is the optimum (e.g. \(J(z) = \operatorname{Re}(z)\)) and initializes the gradient descent at a sufficiently large positive value of \(\theta\), it will run off to positive infinity, instead of going around the circle to reach the global optimum.

More generally, this behavior may also occur with a composite ansatz. Consider the case where the surjective part \(W\) reaches a singular point.
As the \(V\) part is locally surjective by construction, we can still move around the special unitary group in any direction. But we cannot guarantee that we will ever hit a point at which the Riemannian gradient regains support over the derivatives of \(W\), meaning that the parameters of \(W\) might never get updated again. If this does happen, the unitary effectively becomes constraint to a small region in the manifold around \(W(\vec{\theta})\) and the parameters of \(V\) must escape to infinity when the optimum lies outside that region.

If the parameterization was periodic with respect to all parameters, one could argue that the gradient descent effectively never leaves one of the periods and hence, cannot run off to infinity.
In the \(U(1)\) example discussed above, such a periodic parameterization is given by $U(\theta) = e^{i \theta}$.
The corresponding euclidean cost function looks like a sine function, for which  gradient descent surely finds a global optimum when the step size is chosen small enough. The product-of-exponentials ansatz can in principle
exhibit periodicity for a suitably chosen basis, such as the Pauli basis. However it remains unknown if it ever becomes locally surjective.

Convergence of gradient descent can be ensured when the cost function $J$ exhibits compact sublevel sets, i.e., the sets \(\{\vec{\theta} \in \mathds{R}^M \mid J(\vec{\theta}) \leq C\}\) for all $C$ should be compact.
However the VQE cost function given by the expectation value of some finite dimensional Hamiltonian \eqref{eq:VQE_cost} cannot have this property as it is bounded.
One can in principle try to enforce the compact sublevel set property by regularizing the cost function.
For example, one could add a penalty term
\begin{equation}
    J'(\vec{\theta}) = J(\vec{\theta}) + \lambda \lVert \vec{\theta} \rVert_\alpha^\alpha
\end{equation}
with regularization constant \(\lambda\) and order \(\alpha\).
In the machine learing community this strategy is known for \(\alpha = 1\) as L1-regularization while for \(\alpha = 2\) as L2-regularization  \cite{moradi_survey_2020}.

However, regularizing the cost function comes with its own challenges.
Regardless of how small the regularization constant \(\lambda\) is chosen, it cannot be guaranteed anymore  that the L2-regularized cost function \(J'\) is still free of local minima. 
The reason for this is that the negative eigenvalues of the Euclidean Hessian at the critical points can in principle become arbitrarily close to zero.
Any additional positive curvature introduced by the L2-regularization penalty can transform the strict saddles of \(J\) to a local minimum of \(J'\).
Moreover, adding the penalty term also shifts the position of the optimum itself.
The critical points \(\vec{\theta}^*\) of \(J'\) now satisfy \(\nabla J (\vec{\theta}^*) = -\lambda \lVert\vec{\theta}\rVert_2^{-1} \vec{\theta}\) or \(\nabla J (\vec{\theta}^*) = - 2 \lambda \vec{\theta}\) for L1- and L2-regularization, respectively.
It is a priori unclear how far these points, and more importantly the respective states \(\ket{\psi(\vec{\theta^*})}\) and unitaries \(U(\vec{\theta}^*)\), are from their corresponding counterparts of \(J\).

In general, it is possible that the termination of the gradient descent and convergence to the global optimum cannot both be guaranteed at the same time as  some \gls{VQE} instances may not be computable.
It was recently shown that the decision problem associated to a \gls{VQE} with the product-of-exponentials ansatz and commuting generators is undecidable if a certain conjecture on the existence of solutions to a system of polynomial equations holds true \cite{korpas_undecidable_2025}. In the context of this work this means that we would not be able to decide the halting problem for general \gls{VQE} instances.
It may therefore be necessary to make further assumptions on either the parameterized unitary transformation, or the problem Hamiltonian, to achieve guarantees for termination of the VQE and convergence to an optimal solution.
\\[1em]
\indent \textbf{Concluding remarks:} Although \glspl{VQA} have been studied quite extensively in recent times, their convergence properties still remain relatively poorly understood. The results presented in this work provide a first step towards establishing solution guarantees for VQAs. We showed that under the assumption of local surjectivity, a variational quantum eigensolver updated by gradient descent almost never gets stuck at a suboptimal solution.
Instead, the VQE either escapes to infinity or converges towards a global optimum that corresponds to a ground state of a Hamiltonian. 
The assumption of local surjectivity is central to avoid the existence of local minima in the optimization landscape.
Missing directions in the tangent space can lead to the formation of local optima where a gradient-based optimization algorithm might get stuck.

While satisfying local surjectivity does indeed prevent this from happening, it is not an ``easy'' condition to fulfill.
We showed that many of the commonly used parameterized quantum circuit ans\"atze are not locally surjective. In light of these challenges we proposed two alternatives, the family of composite ansätze (Eq.  \eqref{eq:composite_ansätze}) and the Cayley transform (Eq.\eqref{eq:CayleyAnsatz}) that do satisfy local surjectivity. The composite ansätze consist of two parts --- one ensures that local surjectivity is satisfied while the other ensures reachability of all unitaries.
One can build a composite ansatz for example from a \(\mathbb{SU}(d)\)-gate ansatz and a generalized Euler angle parameterization, which can be implemented on quantum hardware \cite{wiersema_here_2024, tilma_generalized_2002}.
The Cayley transform on the other hand cuts the number of required parameters almost in half. The implementation of such parameterized unitary transformation on a quantum device through quantum signal processing techniques is left for future study.

\end{multicols}

\section{Methods}
\label{sec:Methods}

\subsection{Useful results from classical optimization}
\label{sec:classical_optimization}
The convergence properties of gradient descent algorithms have been widely studied in the classical optimization literature. In particular, a lot is known in the case where the cost function is real analytic, which is usually the case for the cost functions we encounter in \glspl{VQA}.
One can show that as long as the step size is chosen so that the descent fulfills the so called \textit{Wolfe conditions}, the gradient search either diverges to infinity or converges to a critical point \(\vec{\theta}^*\) at which the gradient vanishes (cf. Theorem 4.1 in \cite{absil_convergence_2005}),
\begin{align}
    &\lim_{k \rightarrow \infty} \lVert \vec{\theta}_k \rVert = \infty \text{ or } \nonumber \\
    &\lim_{k \rightarrow \infty} \vec{\theta}_k = \vec{\theta}^* , \, \vec{\theta}^* \in \left\{\vec{\theta} \in \mathds{R}^M \mid \nabla J (\vec{\theta}^*) = 0 \right\}.
\end{align}

Such critical points can either be local or global extrema, or saddle points, depending on the eigenvalues of the Hessian of the cost function.
A critical point \(\vec{\theta}^*\) is a strict saddle if the Hessian \(\nabla^2 J(\vec{\theta}^*)\) at this point has at least one negative eigenvalue, i.e. the smallest eigenvalue of the Hessian \(\lambda_\text{min}\) satisfies \(\lambda_\text{min} < 0\) \cite{lee_first-order_2019}.

While the optimizer can in principle get stuck at a strict saddle point, it can be proven from the stable center manifold theorem under the assumption of Lipschitz continuity of the gradient that almost no initialization of gradient descent will actually do so.

\begin{lemma}[cf. Theorem 4 in \cite{pmlr-v49-lee16} and Corollary 2 in \cite{lee_first-order_2019}]
    \label{lem:strict_saddles_are_avoided}
    Let \(J:\mathds{R}^M\to\mathds{R}\) be a twice continuously differentiable function and let \(\mathcal{X}^* = \{ \vec{\theta}^* \in \mathds{R}^M \mid \vec{\theta}^* \text{ is a strict saddle point}\}\) be the set of all strict saddles. Assume that the step size \(\gamma\) fulfills \(0 < \gamma < \frac{1}{L}\), where \(L\) is the Lipschitz constant satisfying \(\lVert \nabla J (\vec{\theta}_1) - \nabla J (\vec{\theta}_2) \rVert_2 \leq L \lVert \vec{\theta}_1 - \vec{\theta}_2 \rVert_2\) for all \(\vec{\theta}_1, \vec{\theta}_2 \in \mathds{R}^M\). Then
    \begin{equation}
        \mathrm{Pr}(\lim_{k \rightarrow \infty} \vec{\theta}_k \in \mathcal{X}^*) = 0
    \end{equation}
    for random initial points \(\vec{\theta}_0\), where the probability is to be taken with respect to an arbitrary measure that is absolutely continuous with respect to the Lebesgue measure.
\end{lemma}

This gives a roadmap to prove the main theorem: establish that the cost function only has global extrema and strict saddle points, then use the above results to guarantee that the strict saddles are avoided and the algorithm can only converge to a global optimizer or escape to infinity.

\subsection{Details on the proof of Theorem \ref{thm:su(d)_convergence}}
\label{sec:proof}
We first give an outline of the proof.
The proof starts by establishing that the Euclidean gradient only vanishes under the assumption of local surjectivity if the Riemannian gradient vanishes.
Next, we analyze the Hessian to show that the Riemannian gradient only vanishes at global optima and strict saddles.
Then, we relate the Riemannian Hessian to the Hessian in the Euclidean parameter space to show that the same also holds there under the assumption of local surjectivity. 
Finally, this enables us to make use of the results in \Cref{sec:classical_optimization} to argue that the algorithm either runs off to infinity or converges to a ground state almost surely.

We start by considering a gradient component
\begin{equation}
    \label{eq:cost_function_gradient}
    \frac{\partial}{\partial \theta_j} J(\vec{\theta}) = \langle \operatorname{grad}J[U(\vec{\theta})], \Omega_j(\vec{\theta}) \rangle,
\end{equation}
where the scalar product is meant to be the Hilbert-Schmidt inner product.
The Riemannian gradient of equation \eqref{eq:VQE_cost} is given by \(\operatorname{grad}J[U(\vec{\theta})] = [H, \ket{\psi(\vec{\theta})} \bra{\psi(\vec{\theta})}]\). Under the assumption that the \(\Omega_j(\vec{\theta})\)'s span \(\mathfrak{su}(d)\), the Euclidean gradient \(\nabla J(\vec{\theta})\) can only be zero at points \(\vec{\theta}^*\) where the Riemannian gradient \(\operatorname{grad}J[U(\vec{\theta}^*)]\) vanishes (\(\operatorname{grad} J[U(\vec{\theta})] \in \mathfrak{su}(d)\) since \(J\) is phase invariant). Consequently, the states \(\ket{\psi(\vec{\theta}^*)}\) created at the critical points need to satisfy \([H, \ket{\psi(\vec{\theta^*})} \bra{\psi(\vec{\theta^*})}] = 0\). However, since \(H\) is diagonalizable, a projector commutes with \(H\) if and only if it projects into an eigenspace of \(H\). Therefore, the states \(\ket{\psi(\vec{\theta}^*)}\) must be eigenstates of \(H\).

We proceed by calculating the Hessian elements at the critical points
\begin{equation}
    \frac{\partial^2}{\partial \theta_i \partial \theta_j} J(\vec{\theta}^*) = - \mathrm{Tr}\{[\Omega_j, [\Omega_i, \ket{\psi(\vec{\theta}^*)} \bra{\psi(\vec{\theta}^*)}]] H\}
\end{equation}
omitting here the explicit dependence of the \(\Omega_j\)'s on \(\vec{\theta}^*\).
In order to determine the type of critical point (that is, the eigenvalue structure of the Hessian at a critical point), we parameterize \(U(\vec{\theta}) = U\) in the neighborhood of \(U(\vec{\theta}^*)\) by \(U = e^{\vec{x} \cdot \vec{B}} U(\vec{\theta}^*)\). Here we use the short hand notation \(\vec{x} \cdot \vec{B} = \sum_{i=1}^{d^2} x_i B_i\) where \(\{B_i \mid i=1,...,d^2\}\) is a complete and orthonormal basis for \(\mathfrak{u}(d)\). In the neighborhood of \(U(\vec{\theta}^*)\), the cost function is given by
\begin{align}
    J_{\{B_i\}}(\vec{x}) &= \mathrm{Tr}\{e^{\vec{x} \cdot \vec{B}} \ket{\psi(\vec{\theta}^*)}\bra{\psi(\vec{\theta}^*)}e^{-\vec{x}\cdot\vec{B}}H\} \\
    & = J(\vec{\theta}^*) + \mathrm{Tr}\{(\vec{x}\cdot\vec{B})^2\ket{\psi(\vec{\theta}^*)}\bra{\psi(\vec{\theta}^*)} H\} -\mathrm{Tr}\{(\vec{x}\cdot\vec{B})\ket{\psi(\vec{\theta}^*)}\bra{\psi(\vec{\theta}^*)}(\vec{x}\cdot\vec{B})H\} + \mathcal{O}(\lVert x \rVert_2^3), \nonumber
\end{align}
where we used that the gradient vanishes at a critical point and the subscript \(\{B_i\}\) indicates the basis used. Now, taking \(\{B_i\} = \{d^{-\frac{1}{2}}\mathds{1},\lambda_s, \lambda_{k,l}, \bar{\lambda}_{k,l}\}, 1 \leq s \leq d-1, 1 \leq k < l \leq d\), where
\begin{align}
    &\lambda_s = \frac{\mathrm{i}}{\sqrt{s(s+1)}}\left(\sum_{r=1}^s \ket{E_r}\bra{E_r} -s\ket{E_{s+1}}\bra{E_{s+1}}\right) \\
    & \lambda_{k,l} = \frac{1}{\sqrt{2}} \mathrm{i}(\ket{E_k}\bra{E_l} + \ket{E_l}\bra{E_k}), \\
    & \bar{\lambda}_{k, l} = \frac{1}{\sqrt{2}}(\ket{E_k}\bra{E_l} - \ket{E_l}\bra{E_k})
\end{align}
are the generalized Gell-Mann matrices and \(\ket{E_j}\) are eigenstates of \(H\) corresponding to the eigenenergy \(E_j\) for \(j=1,...,d\). W.l.o.g. we can assume that \(\ket{\psi(\vec{\theta^*})} = \ket{E_p}\) for some index \(p\).
We find (see also Theorem IV.2 in \cite{wang_analysis_2010}) that the Hessian \(\nabla^2J_{\{d^{-\frac{1}{2}}\mathds{1},\lambda_s,\lambda_{k,l}, \bar{\lambda}_{k,l}\}}(0)\) at a critical point \(\vec{x}=0\) is diagonal.
Furthermore, the first \(d\) diagonal entries corresponding to the identity and \(\lambda_s\) vanish. The remaining \(d^2-d\) Hessian diagonal elements in the subspace spanned by \(\{\lambda_{k,l}, \bar{\lambda}_{k, l}\}\), read
\begin{align}
\label{eq:diagonal_hessian}
\frac{\partial^2}{\partial x_i^2} J_{\{\lambda_{k,l}, \bar{\lambda}_{k,l}\}}(0) & = 2(\mathrm{Tr}\{B_i^2\ket{\psi(\vec{\theta}^*)}\bra{\psi(\vec{\theta}^*)} H\} - \mathrm{Tr}\{B_i \ket{\psi(\vec{\theta}^*)} \bra{\psi(\vec{\theta}^*)} B_i H\}) \nonumber \\
    &= \frac{1}{2}(E_k - E_l)(\lvert \braket{\psi(\vec{\theta}^*) | E_l}\rvert^2 - \lvert \braket{\psi(\vec{\theta}^*) | E_k}\rvert^2),
\end{align}
for some \(B_i \in \{\lambda_{k, l}, \bar{\lambda}_{k,l}\}\). Now if we assign the eigenvalue \(E_1\) to the ground state, i.e. \(E_1 \leq E_2 \leq \cdots \leq E_d\), the critical points \(\vec{\theta}^*_\mathrm{min}\) that correspond to a ground state constitute Hessian eigenvalues that are either 0 or positive, i.e.
\begin{align}
    \frac{\partial^2}{\partial x_i^2} J_{\{\lambda_{k,l}, \bar{\lambda}_{k,l}\}}(0) & = \frac{1}{2}(E_k - E_l)(\lvert \braket{\psi(\vec{\theta}^*_\mathrm{min}) | E_l}\rvert^2 - \lvert \braket{\psi(\vec{\theta}^*_\mathrm{min}) | E_k}\rvert^2) \nonumber \\
    & = \begin{cases}
        0, & \text{for } E_k = E_l = E_1 \\
        \frac{1}{2}(E_l - E_1) > 0 & \text{for } E_l > E_k = E_1 \\
        0 & \text{for } E_k > E_1.
    \end{cases}
\end{align}
Conversely, if we assign \(E_d\) to the largest eigenvalue of \(H\), the corresponding Hessian eigenvalues are either 0 or negative. As long as \(E_d > E_1\), all intermediate cases have at least one negative Hessian eigenvalue. In the case of \(E_d = E_1\), \(H\) was a multiple of the identity and therefore any state is a ground state of \(H\). Note that since these arguments hold independently of whether \(B_i \in \{\lambda_{k, l}\}\) or \(B_i \in \{\bar{\lambda}_{k,l}\}\), each Hessian eigenvalue must appear twice. Moreover, in the \(2(d-g)\) dimensional subspace, where \(g\) is the degeneracy of the ground state, spanned by \(\{\lambda_{1, p}, \bar{\lambda}_{1, p} \mid p > g\}\), the Hessian eigenvalues are non-zero at the critical points. This means that within this subspace there are no singular critical points at which the Hessian is not invertible, which shows that the cost function \(J(\vec{x})\) is a Morse function within that subspace.

In order to show that all saddles are strict, i.e. in order to establish our theorem, we need to relate the Hessian \(\nabla^2 J_{\{B_i\}}(0)\) to the Euclidean Hessian \(\nabla^2 J (\vec{\theta}^*)\) at the critical points. In order to do so it suffices to consider only the non-zero subspace spanned by \(\{\lambda_{1, p}, \bar{\lambda}_{1, p} \mid p > g\}\) of \(\nabla^2 J_{\{B_i\}}(0)\), where we denote by \(E_p\) the eigenvalue that corresponds to the saddle point we consider. We go on and expand the operators \(B_i\) in equation \eqref{eq:diagonal_hessian} in the basis \(\{\Omega_j \mid j=1,...,d^2-1\}\) of \(\mathfrak{su}(n)\) to obtain
\begin{align}
    \frac{\partial^2}{\partial x_i^2} J_{\{B_i\}}(0) & = \sum_{n,m=1}^{d^2-1}S_{n, i} S_{m,i} (\mathrm{Tr}\{(\Omega_n \Omega_m + \Omega_m \Omega_n)\ket{\psi(\vec{\theta}^*)} \bra{\psi(\vec{\theta}^*)} H\} \nonumber \\
    & \hphantom{= \sum_{n,m=1}^{d^2-1}} - \mathrm{Tr} \{(\Omega_n \ket{\psi(\vec{\theta}^*)}\bra{\psi(\vec{\theta}^*)}\Omega_m
     + \Omega_m \ket{\psi(\vec{\theta}^*)} \bra{\psi(\vec{\theta}^*)} \Omega_n) H\}) \nonumber \\
     &=  \sum_{n,m=1}^{d^2-1} S_{n,i} S_{m,i} \mathrm{Tr}\{[\Omega_m, [\Omega_n, \ket{\psi(\vec{\theta}^*)}\bra{\psi(\vec{\theta}^*)}]]H\} \\
    &= \sum_{n, m=1}^{d^2-1} S_{n,i} S_{m,i} \frac{\partial^2}{\partial\theta_n\partial\theta_m}J(\vec{\theta}^*), \nonumber
\end{align}
where the coefficients are given by the inner product \(S_{n,i}=\mathrm{Tr}\{B_i^\dagger \Omega_n\}\).
Thus, the Hessians at the critical points are related through the transformation
\begin{equation}
    \label{eq:basis_change}
    \nabla^2 J_{\{B_i\}}(0) = S^T\nabla^2 J(\vec{\theta}^*) S,
\end{equation}
where \(S_{n,i}\) are the elements of the matrix \(S\).

It remains to show that \(\nabla^2 J (\vec{\theta^*})\) has at least one negative eigenvalue, whenever \(\nabla^2 J_{\{B_i\}}(0)\) does.
Since J depends smoothly on \(\vec{\theta}\), the Hessian \(\nabla^2 J (\vec{\theta^*})\) is symmetric.
Denote by \(\sigma_k(A)\) the k-th eigenvalue of a matrix \(A\) in ascending order.
A generalization of Sylvester's law of inertia to rectangular transformations (see Theorem 3.2 in \cite{higham_modifying_1998}) states that
\begin{equation}
    \sigma_1(\nabla^2 J_{\{B_i\}}(0)) = \gamma_1 \mu_1,
\end{equation}
where \(\sigma_1(S^TS) \leq \gamma_1\) and \(\sigma_1(\nabla^2 J (\vec{\theta^*})) \leq \mu_1\).
Since the \(\Omega_j\) span \(\mathfrak{su}(d)\), we have that \(S\) is full rank and therefore all eigenvalues of \(S^T S\) are strictly positive.
Hence, \(\gamma_1\) is positive and \(\mu_1\) is negative whenever \(\nabla^2 J_{\{B_i\}}(0)\) has at least one negative eigenvalue.
It follows that in that case, \(\nabla^2 J (\vec{\theta^*})\) also has at least one negative eigenvalue, which ensures that the Euclidean optimization landscape also only has global extrema and saddle points.
Since the Lie algebra elements \(\Omega_j(\vec{\theta})\) are uniformly bounded by assumption, so is the Euclidean Hessian and hence the gradient of the cost function is Lipschitz continuous.
Therefore, we can apply \Cref{lem:strict_saddles_are_avoided} to argue that strict saddle points are almost always avoided, which completes the proof.

\section*{Declarations}
\subsection*{Funding}
The research is part of the Munich Quantum Valley, which is supported by the Bavarian state government with funds from the Hightech Agenda Bavaria Plus. C.A. acknowledges support from the National Science Foundation (NSF Grant No. 2231328).  

\subsection*{Competing Interests Statement}
There are no competing interests.

\subsection*{Author contributions}
M.W. constructed the locally surjective ansätze, formalized the results and co-wrote the manuscript. D.B., G.D., T.S.-H. and E.M. provided key ideas, insights and feedback. G.D. constructed the modified Cayley transform. C.A. contributed the initial idea for the project, parts of the main proof and co-wrote the manuscript. The project was jointly supervised by C.A. and D.B.

\printbibliography

\appendix
\section*{Appendix}

\section{Quantum control versus VQAs}
\label{sec:QC_vs_QVAs}

Given a (bilinear) quantum control system
\begin{equation}
\label{eq:bil-sys}
\dot{U}(t) = - {\rm i}H(f(t))U(t)\,, \quad U(0) = I_n
\end{equation}
evolving on the (special) unitary group, where $H(f)$ denotes a control Hamiltonian depending on certain input parameters. 
In many applications, $H(f)$ takes the affine form 
$H(f) := H_0 + \sum_{j=1}^m f_j H_j$ with $H_0$ called drift Hamiltonian
and $H_j$ control Hamiltonians.

Now fix a terminal time $T >0$ and introduce the endpoint map 
$f \mapsto E_T(f) \in SU(n)$ defined by
\begin{equation}
E_T(f) := U_f(T)\,,
\end{equation}
where $U_f(t)$ denotes the unique solution of \eqref{eq:bil-sys} 
corresponding to the control $f:[0,T] \to \mathbb{R}^m$. Here, we assume that $f$ is chosen from an appropriate function space $F$, e.g. $F = L^1([0,T],\mathbb{R}^m)$. Next consider a smooth map $\vec{\theta} \mapsto p(\vec{\theta})$ from $\mathbb{R}^M$ to $F$.
Then the restricted endpoint map $\mathbb{R}^M \to E_T(p(\vec{\theta}))$, which results from restricting the endpoint map to a ``finite dimensional subset'' of the space of control functions, leads to a parametrization
\begin{equation}
U(\vec{\theta}):= E_T(p(\vec{\theta}))\,.
\end{equation}
In this sense quantum control systems give rise to unitary parameterizations and hence can be regarded as a \gls{VQA} \cite{PRXQuantum.2.010101}.

\section{Uniform boundedness of the Lie-algebra elements}
\label{sec:lipschitz}
In order to apply \Cref{lem:strict_saddles_are_avoided}, we need to make sure that the gradient of the cost function is Lipschitz continuous.
We will show in the next section that this is the case when the Lie algebra elements of the parameterization are uniformly bounded, i.e. if for all \(\vec{\theta} \in \mathds{R}^M\) and all \(j=1,...,M\) we have \(\lVert \Omega_j(\vec{\theta}) \rVert \leq C\) for some constant \(C\), which does not depend on \(\vec{\theta}\) or \(j\).
In this section, we show that this is indeed the case for the family of composite ansätze (cf. equation \eqref{eq:composite_ansätze}) and the Cayley transform.

We begin with the family of composite ansätze.
Since by construction, \(\Omega_1(\vec{\phi}), ..., \Omega_{d^2-1}(\vec{\phi})\) do not depend on \(W\) and \(\vec{\theta}\), we can analyze \(V\left(\frac{2R}{\pi} \arctan(\vec{\phi})\right)\) in isolation.
There are two cases: either \(V\) is a product-of-exponentials ansatz, in which case \(\lVert \Omega_j(\vec{\phi})\rVert \leq \lVert X_j \rVert\) in any unitarily invariant norm.
Hence, we can choose \(C = \max_j \lVert X_j \rVert\) and get uniform boundedness.
Or \(V\) is the \(\mathbb{SU}(d)\)-gate ansatz, in which case we can argue that the \(\lVert \Omega_j(\vec{\phi}) \rVert\) is a continuous function, hence it has a finite maximal value on the compact disc with radius \(R\) around the origin.
Since the composite ansatz only evaluates \(V\) inside that disc and the arctangent only contributes a factor smaller than one to the derivative, we again obtain uniform boundedness.
For the remaining \(\Omega_{d^2}(\vec{\theta}, \vec{\phi}), ..., \Omega_{2(d^2-1)}(\vec{\theta}, \vec{\phi})\) just note that they are given by partial derivatives of \(W(\vec{\theta})\) up to some unitary factors.
Since the partial derivatives are assumed to be uniformly bounded per definition, the same will hold for the \(\Omega s\) in an appropriate unitarily invariant norm.

Finally, for the Cayley transform the differential is given by (see Lemma 8.8 in \cite{hairer_geometric_2006})
\begin{align}
    \mathrm{d} \mathrm{cay}_X: \mathfrak{u}(d) &\rightarrow T_{\mathrm{cay}(X)}\mathrm{U}^*(d) \\
    H &\mapsto 2(\mathds{1}-X)^{-1} H (\mathds{1} + X)^{-1},
\end{align}
we get that
\begin{equation}
    \frac{\partial}{\partial \theta_j} \mathrm{cay}(X(\vec{\theta})) = 2(\mathds{1} - X(\vec{\theta}))^{-1}X_j(\mathds{1} + X(\vec{\theta}))^{-1}.
\end{equation}
And hence
\begin{equation}
    \Omega_j(\vec{\theta}) = \mathrm{cay}(X(\vec{\theta}))^\dagger \frac{\partial}{\partial \theta_j} \mathrm{cay}(X(\vec{\theta})) = 2(\mathds{1} + X(\vec{\theta}))^{-1} X_j (\mathds{1} + X(\vec{\theta}))^{-1},
\end{equation}
where we used that \(\mathrm{cay}(X)^\dagger = \mathrm{cay}(-X)\).
Since the spectrum of \(X(\vec{\theta})\) is purely imaginary, we have for all eigenvalues \(\lambda\) of \((\mathds{1} + X(\vec{\theta}))\) that \(\lvert \lambda \rvert \geq 1\).
In particular, \((\mathds{1} + X(\vec{\theta}))\) is diagonalizable.
Hence the operator norm fulfills \(\lVert (\mathds{1} + X(\vec{\theta}))^{-1} \rVert_\infty \leq 1\) and we get uniform boundedness by submultiplicativity of the operator norm.

\section{Relaxing the analyticity requirement of the parameterization}
\label{sec:morse-bott}
In the main Theorem, we require the parameterization to be analytic.
This is necessary so that we can apply Theorem 4.1 of \cite{absil_convergence_2005} to guarantee that the gradient descent converges to a single, critical point.
In general, smoothness of the loss function is not sufficient for this guarantee, see e.g. Section 3.2.1. in \cite{absil_convergence_2005}.
However, for continuous gradient flows on compact manifolds, it is known that for the large class of Morse-Bott functions, convergence to a single point can be guaranteed (see e.g. Appendix A in \cite{marinkovic_symplectic_2022} for a detailed proof).
Here, we provide substantial steps necessary to apply this result to our setting.
Some details are however still missing, so this is not a complete, rigorous proof.

Let us define the notion of a Morse-Bott function (cf. e.g. \cite{banyaga_proof_2004})
\begin{definition}
    Let N be a smooth manifold and \(f: N \to \mathds{R}\) be a smooth function.
    Then \(f\) is said to be Morse-Bott, if
    \begin{enumerate}
        \item The set of critical points of \(f\) is a disjoint union of connected, smooth submanifolds \(C_k\) such that f is constant on each component.
        \item For each \(p \in C_k\) the kernel of the Hessian of \(f\) at \(p\) is exactly the tangent space \(T_pC_k\) of \(p\) in \(C_k\).
    \end{enumerate}
\end{definition}
First, we are going to argue that the cost functional \(J[U]\) is Morse-Bott.
The critical points of \(J[U]\) are exactly the unitaries \(U\) s.t. \(U\ket{\psi_0}\) is an eigenstate of \(H\).
For each eigenspace \(V_k, k=1,...,n_H\) of \(H\), where \(n_H\) denotes the number of distinct eigenvalues of \(H\), denote with \(g_k = \operatorname{dim} V_k\) its degeneracy.
Then we have for the critical points \(\operatorname{Crit }(J[U])\) the following decomposition:
\begin{equation}
    \operatorname{Crit}(J[U])= \bigsqcup_{k=1}^{n_H} \{U \in \mathrm{U}(d) \ | \ U\ket{\psi_0} \in V_k\} \text{\reflectbox{$\coloneq$}} \bigsqcup_{k=1}^{n_H} C_k.
\end{equation}
The map \(\mathrm{U}(d) \to \mathcal{S}^{2d-1} \text{; } U \mapsto U \ket{\psi_0}\), where \(\mathcal{S}^{2d-1}\) is the \(2d-1\)-dimensional unit sphere, is a smooth submersion.
Hence the preimage of \(\mathcal{S}(V_k)\), the unit sphere in \(V_k\), is a smooth submanifold of \(\mathrm{U}(d)\) of dimension \(d^2 - (2d-1 - (2g_k-1)) = d^2-2(d-g_k)\).
But this preimage is exactly \(C_k\), which proves the first of the Morse-Bott conditions.
Moreover, we already saw in \Cref{sec:proof} that the Hessian of \(J[U]\) has exactly \(d^2 - 2(d-g_k)\) zeros at any critical point belonging to \(C_k\).
Hence \(J[U]\) must be Morse-Bott.

Next, we need to argue that as long as the parameterization \(U(\vec{\theta})\) satisfies local surjectivity, then the cost function \(J(\vec{\theta})\) on the Euclidean parameter space is also Morse-Bott.
To see this, note that local surjectivity implies that the parameterization is a (smooth) submersion.
Hence the preimages of the critical submanifolds \(C_k\) are themselves again smooth submanifolds whose dimension increases by exactly \(M-d^2\).
However as a direct consequence of \cref{eq:basis_change} and Theorem 3.2 in \cite{higham_modifying_1998}, the number of non-zero eigenvalues of the Hessian of \(J(\vec{\theta})\) is the same as the one of the Hessian of \(J[U]\), when the parameterization fulfills local surjectivity.
Hence, \(J(\vec{\theta})\) is also Morse-Bott.

For any continuous gradient flow, i.e. solution to the differential equation \(\dot{\vec{\theta}}(t) = -\nabla J(\vec{\theta}(t))\), that does not run off to infinity, we can restrict our attention to some compact subset of \(\mathds{R}^M\).
Then, according to Theorem 2.2 in \cite{marinkovic_symplectic_2022}, the gradient flow must converge to a single, critical point of \(J(\vec{\theta})\).
We expect that a discrete gradient descent algorithm (cf. \cref{eq:gradient_descent}) follows the continuous flow close enough to also converge to a single, critical point if the step size is chosen small enough.

\section{Overparemeterization does not remove singularities of the SU(d) gate ansatz}
\label{sec:overparameterization}
The singular points of the \(\mathbb{SU}(d)\)-gate ansatz are well known and can be studied in detail.
In particular, it is possible to determine which directions in the tangent space are actually missing at these singular points.
This knowledge, captured in the following Lemma, will allow us to argue that the natural strategies for overparameterization cannot fully eliminate singular points in this case.
\begin{lemma}
    \label{lem:singular_points_exp}
    The \(\mathbb{SU}(d)\)-gate ansatz is not locally surjective.
    Its singular points are given by
    \begin{equation}
        \label{eq:singular_points_first_kind}
        \mathrm{sing}(U(\vec{\theta})) = \left\{\vec{\theta} \in \mathds{R}^{d^2-1} \big| \left\{\lambda_1 - \lambda_2 \mid \lambda_1, \lambda_2 \in \mathrm{spec}\left(X(\vec{\theta})\right)\right\} \cap \left(2 \pi i \mathds{Z} \setminus \{0\}\right) \neq \varnothing \right\}.
    \end{equation}
    Furthermore at any singular point \(\vec{\theta}_s \in \mathrm{sing}(U(\vec{\theta}))\), let \(\ket{e_i}\) be an orthonormal eigensystem of \(X(\vec{\theta}_s)\) with corresponding eigenvalues \(\lambda_i\) for \(i=1,...,d^2-1\). We have that
    \begin{equation}
        \mathrm{span}\{\Omega_j(\vec{\theta}_s)\}^\perp = \mathrm{span}\{i(\ket{e_i}\bra{e_j} + \ket{e_j}\bra{e_i}),\ket{e_i}\bra{e_j} - \ket{e_j}\bra{e_i} \mid  \lambda_i - \lambda_j \in 2\pi i \mathds{Z} \setminus \{0\}\},
    \end{equation}
    where the orthogonal complement is taken with respect to \(\mathfrak{su}(d)\).
\end{lemma}
\begin{proof}
    For the characterization of the singular points, also confer to page 313 of \cite{hilgert_structure_2012}.
    The differential of the exponential map is given by (cf. Theorem 1.7 in \cite{sigurdur_differential_1978})
    \begin{align}
    \label{eq:exp_jac}
        \mathrm{d}U_{\vec{\theta}}: \mathds{R}^{d^2-1} &\rightarrow T_{\exp(X(\vec{\theta}))} \mathrm{SU}(d) \nonumber \\
        \vec{y} & \mapsto \mathrm{d}\left(L_{\exp(X(\vec{\theta}))}\right)_e \circ \frac{1-\exp(-\mathrm{ad}_{X(\vec{\theta})})}{\mathrm{ad}_{X(\vec{\theta})}}( X(\vec{y}))
    \end{align}
    where \(L_{\exp(X(\vec{\theta}))}\) is the left multiplication by \(\exp(X(\vec{\theta}))\) on \(\mathrm{SU}(d)\), \(\mathrm{ad}_{X(\vec{\theta})}\) is the adjoint action of \(X(\vec{\theta})\) on \(\mathfrak{su}(d)\) and \(\frac{1-\exp(-\mathrm{ad}_{X(\vec{\theta})})}{\mathrm{ad}_{X(\vec{\theta})}}\) has to be understood as its power series expansion
    \begin{equation}
        \frac{1-\exp(-\mathrm{ad}_{X(\vec{\theta})})}{\mathrm{ad}_{X(\vec{\theta})}} = \sum_{k=0}^\infty \frac{(-1)^k \mathrm{ad}_{X(\vec{\theta})}^k}{(k+1)!}.
    \end{equation}
    Since \(L_{\exp(X(\vec{\theta}))}\) is a diffeomorphism on \(\mathrm{SU}(d)\), \(\mathrm{d}\exp_{\vec{\theta}}(\vec{y}) = 0\) if and only if
    \begin{equation}
        X(\vec{y}) \in \text{ker}\left(\frac{1-\exp(-\mathrm{ad}_{X(\vec{\theta})})}{\mathrm{ad}_{X(\vec{\theta})}}\right).
    \end{equation}
    This kernel can be calculated to be (cf. page 316 of \cite{hilgert_structure_2012})
    \begin{equation}
    \label{eq:kernel_exp_diff}
        \mathrm{ker}\left(\frac{1-\exp(-\mathrm{ad}_{X(\vec{\theta})})}{\mathrm{ad}_{X(\vec{\theta})}}\right) = \bigoplus_{z \in I} \mathrm{ker}(\mathrm{ad}_{X(\vec{\theta})} - z \mathds{1}),
    \end{equation}
    where the index set \(I\) is given by \(I = \mathrm{spec}(\mathrm{ad}_{X(\vec{\theta})}) \cap (2\pi i \mathds{Z} \setminus\{0\})\).
    The eigenvalues of \(\mathrm{ad}_{X(\vec{\theta})}\) are exactly \(\lambda_i - \lambda_j\), which proves equation \eqref{eq:singular_points_first_kind}.

    To calculate the orthogonal complement to the span of the \(\Omega_j(\vec{\theta}_s)\) at a singular point \(\vec{\theta}_s\), we use that the orthogonal complement to the range of a linear map is the kernel of its adjoint.
    For any \(Y \in T_{\exp(X(\vec{\theta}))}\mathrm{SU}(d)\) and \(\vec{z} \in \mathds{R}^{d^2-1}\), we have that
    \begin{align}
        \langle Y, \mathrm{d}U_{\vec{\theta}}(\vec{z})\rangle & = \sum_{k=0}^\infty \frac{(-1)^k}{(k+1)!} \mathrm{Tr}\{Y^\dagger\exp(X(\vec{\theta}))\mathrm{ad}_{X(\vec{\theta})}^k(X(\vec{z}))\}
         \nonumber \\
        & = \sum_{k=0}^\infty \frac{1}{(k+1)!} \mathrm{Tr} \{\mathrm{ad}_{X(\vec{\theta})}^k(Y^\dagger \exp(X(\vec{\theta})))X(\vec{z})\} \nonumber  \\
        & = \sum_{k=0}^\infty \frac{1}{(k+1)!} \mathrm{Tr}\left\{\left[\mathrm{ad}_{X(\vec{\theta})}^k(\exp(-X(\vec{\theta}))Y)\right]^\dagger X(\vec{z}) \vphantom{\left[\mathrm{ad}_{X(\vec{\theta})}^k(\exp(-X(\vec{\theta}))Y)\right]^\dagger} \right\} \nonumber \\
        & = \left\langle X^* \circ \frac{\exp(\mathrm{ad}_{X(\vec{\theta})}) - 1}{\mathrm{ad}_{X(\vec{\theta})}} \circ \mathrm{d}\left(L_{\exp(-X(\vec{\theta}))}\right)_e(Y), \vec{z} \vphantom{\frac{\exp(\mathrm{ad}_{X(\vec{\theta})}) - 1}{\mathrm{ad}_{X(\vec{\theta})}}} \right\rangle,
    \end{align}
    where \(X^*\) is the adjoint of \(X\).
    Hence,
    \begin{equation}
        \mathrm{d}U_{\vec{\theta}}^* = X^* \circ \frac{\exp(\mathrm{ad}_{X(\vec{\theta})}) - 1}{\mathrm{ad}_{X(\vec{\theta})}} \circ \mathrm{d}\left(L_{\exp(-X(\vec{\theta}))}\right)_e\
    \end{equation}
    Since the \(X_j\) were assumed to be a basis of \(\mathfrak{su}(d)\), \(\mathrm{ker}(X^*) = \{0\}\). Therefore,
    \begin{equation}
            \mathrm{ker}(\mathrm{d}U_{\vec{\theta}}^*) = \left\{\vphantom{\Omega \in \mathrm{ker}\left(\frac{\exp(\mathrm{ad}_{X(\vec{\theta})}) - 1}{\mathrm{ad}_{X(\vec{\theta})}}\right)} \exp(X(\vec{\theta})) \Omega \mid \Omega \in \mathrm{ker}\left(\frac{\exp(\mathrm{ad}_{X(\vec{\theta})}) - 1}{\mathrm{ad}_{X(\vec{\theta})}}\right)\right\}.
    \end{equation}
    The kernel is the same as the one given in equation \eqref{eq:kernel_exp_diff}, which concludes the proof.
\end{proof}

For the parameterization given in \Cref{def:canonical_coordinates_first_kind}, two natural strategies for overparameterizations are given by
\begin{enumerate}
    \item Using an overcomplete frame \(\{X_j \mid j = 1,...,M > d^2-1\}\) of \(\mathfrak{su}(d)\) to construct the Lie algebra element \(X(\vec{\theta}) = \sum_{j=1}^M\theta_j X_j\).
    \item Multiplying multiple copies of the ansatz, e.g.
    \begin{equation}
        V(\vec{\theta}, \vec{\varphi}) = U(\vec{\theta}) U(\vec{\varphi}).
    \end{equation}
\end{enumerate}
The following theorem proves that these types of overparameterization also always admit singular points.

\begin{samepage}
\begin{lemma}
\noindent
    \begin{itemize}
        \item[(a)] Let \(X_j\), \(j=1,...,M > d^2-1\) be an overcomplete frame of \(\mathfrak{su}(n)\) and let
        \begin{equation*} 
            W(\vec{\theta}) = \exp(\sum_{j=1}^M \theta_j X_j).
        \end{equation*}
        Then this parameterization is not locally surjective, i.e. it has singular points.
        \item[(b)] Let \(X_j\), \(j=1,...,d^2-1\) be a basis of \(\mathfrak{su}(d)\). Further let \(\vec{\theta}, \vec{\varphi} \in \mathds{R}^{d^2-1}\) and 
        \begin{equation*}
            V(\vec{\theta}, \vec{\varphi}) = \exp(X(\vec{\theta})) \exp(X(\vec{\varphi})).
        \end{equation*}
        If \(\vec{\varphi}_s\) is a singular point of \(\exp(X(\vec{\varphi}))\), then \((\vec{\varphi}_s, \vec{\varphi}_s)\) is a singular point of \(V\).
    \end{itemize}
\end{lemma}
\end{samepage}
\begin{proof}
    \noindent
    \begin{itemize}
        \item [(a)] As can be seen from equation \eqref{eq:exp_jac}, we get two contributions to the kernel of the differential of the parameterization.
        First, there is a \(M - (d^2-1)\)-dimensional subspace for which \(X(\vec{y})=0\) and hence \(\mathrm{d}U_{\vec{\theta}}(\vec{y}) = 0\), Secondly, for suitable \(\vec{\theta}\) there exists a \(\vec{y}\) such that \(X(\vec{y}) \neq 0\) is in the kernel of \(\frac{1-\exp(-\mathrm{ad}_{X(\vec{\theta})})}{\mathrm{ad}_{X(\vec{\theta})}}\).
        Therefore we have for suitable choices of \(\vec{\theta}\) that
        \begin{align}
            \mathrm{dim}\,\mathrm{im}\left(\mathrm{d}U_{\vec{\theta}}\right) & = M - \mathrm{dim}\,\mathrm{ker}\left(\mathrm{d}U_{\vec{\theta}}\right) \nonumber \\
            & < M - (M-d^2-1) = d^2-1,
        \end{align}
        which shows that the differential is not surjective at those points.
        \item[(b)] We have that
        \begin{align}
            V^\dagger(\vec{\theta}, \vec{\varphi}) \frac{\partial}{\partial \varphi_j}V(\vec{\theta}, \vec{\varphi}) = & \Omega_j(\vec{\varphi}) \\
            V^\dagger(\vec{\theta}, \vec{\varphi}) \frac{\partial}{\partial \theta_j}V(\vec{\theta}, \vec{\varphi}) = & \exp(-X(\vec{\varphi})) \Omega_j(\vec{\theta}) \exp(X(\vec{\varphi})).
        \end{align}
        Furthermore, from \Cref{lem:singular_points_exp} we know that
        \begin{equation}
                \mathrm{span}\{\Omega_j(\vec{\varphi}_s)\} = \mathrm{span}\{i(\ket{e_i}\bra{e_j} + \ket{e_j}\bra{e_i}), \ket{e_i}\bra{e_j} - \ket{e_j}\bra{e_i} \mid  \lambda_i - \lambda_j \in 2\pi i \mathds{Z} \setminus \{0\}\}^\perp,
        \end{equation}
        where \(X(\vec{\varphi}_s)\ket{e_{\{i,j\}}} = \lambda_{\{i,j\}} \ket{e_{\{i,j\}}}\).
        Choose \(i,j\) such that \(\lambda_i - \lambda_j \in 2\pi i \mathds{Z} \setminus \{0\}\).
        Then,
        \begin{align}
            & \langle \exp(-X(\vec{\varphi}_s)) \Omega_k(\vec{\varphi}_s) \exp(X(\vec{\varphi}_s)), i(\ket{e_i}\bra{e_j} + \ket{e_j}\bra{e_i}) \rangle \nonumber \\
            & = \langle \Omega_k(\vec{\varphi}_s), i\exp(X(\vec{\varphi}_s))(\ket{e_i}\bra{e_j} + \ket{e_j}\bra{e_i}) \exp(-X(\vec{\varphi}_s))\rangle \nonumber \\
            & = \langle \Omega_k(\vec{\varphi}_s) , i(e^{\lambda_i - \lambda_j}\ket{e_i}\bra{e_j} + e^{\lambda_j - \lambda_i}\ket{e_j}\bra{e_i}) \rangle \nonumber \\
            & = \langle \Omega_k(\vec{\varphi}_s) , i(\ket{e_i}\bra{e_j} + \ket{e_j}\bra{e_i}) \rangle = 0.
        \end{align}
        And similarly for \(\ket{e_i}\bra{e_j} - \ket{e_j}\bra{e_i}\).
        Therefore, \(\exp(-X(\vec{\varphi}_s)) \Omega_j(\vec{\varphi}_s) \exp(X(\vec{\varphi}_s)) \in \mathrm{span}\{\Omega_j(\vec{\varphi}_s)\}\) and \((\vec{\varphi}_s, \vec{\varphi}_s)\) is a singular point.
    \end{itemize}
\end{proof}

\section{Numerical evidence that overparameterization does not remove singular points of the products-of-exponentials ansatz}
\label{sec:numerical_evidence}
Consider any order of the product of exponentials of \(\sigma_x\), \(\sigma_y\) and \(\sigma_z\) and repeat that ansatz, so that it has the same number of parameters as our composite ansatz.
Consider the Jacobian \(dU(\vec{\theta})\) of the parameterization \(U(\vec{\theta})\).
We construct a cost function \(C(\vec{\theta}) = \det\{dU(\vec{\theta})[dU(\vec{\theta})]^\dagger\} \geq 0\).
The eigenvalues of \(dU(\vec{\theta})[dU(\vec{\theta})]^\dagger\) are exactly the modulus-squares of the singular values of the Jacobian.
A zero of the cost function therefore corresponds to points where one of the singular values is zero and hence the Jacobian is not of full rank.
Therefore we can numerically search for singular points by minimizing this cost function.
We use the \emph{NMinimize} function of the \emph{Mathematica} computer algebra system to do the search and find that no matter the order of the generators, singular points can always be found.

\section{A modified Cayley transform}
\label{sec:modified_cayley}
The Cayley transform
\begin{align}
    \mathrm{cay}: \mathfrak{u}(d) &\rightarrow \mathrm{U}^*(d) \\
    X &\mapsto (\mathds{1}-X)^{-1}(\mathds{1} + X),
\end{align}
where we denote by \(\mathrm{U}^*(d)\) the set of unitaries \(G\), s.t \(\det(G + \mathds{1}) \neq 0\), is not surjective onto \(\mathrm{U}(d)\).
Since its image \(\mathrm{U}^*(d)\) still contains every unitary up to an irrelevant global phase factor, this does not pose a problem for our purposes.
However, we want to remark that the Cayley transform can be modified such that it becomes surjective onto \(\mathrm{SU}(d)\) in the strict mathematical sense.
Consider the following modification
\begin{align}
    \widetilde{\mathrm{cay}}: \mathfrak{u}(d) &\rightarrow \mathrm{SU}(d) \\
    X &\mapsto \left(\operatorname{det}\left(\sqrt[d]{\operatorname{cay}(X)}\right)\right)^{-1} \operatorname{cay}(X).
\end{align}
Here we define the \(d\)-th root of a unitary by acting on its eigenvalues \(\lambda_j\) with
\begin{equation}
    \sqrt[d]{\lambda_j} = \exp\left(\frac{1}{d}\log(\lambda_j)\right),
\end{equation}
where we always take the principal branch of the complex logarithm.
Since the Cayley transform never maps to a unitary with \(-1\) as an eigenvalue, the branch cut of the logarithm is always avoided.
Hence the modified Cayley transform is a smooth map.
Moreover, its image is clearly the full \(\mathrm{SU}(d)\).
Note that it is important to take the \(d\)-th root of the unitary \(\operatorname{cay}(X)\) first and the determinant second.
If the determinant was taken first, there would be no way to ensure that the branch cut is avoided.

Now take any basis \(X_j\) of \(\mathfrak{u}(d)\) and consider \(X(\vec{\theta}) = \sum_{j=1}^{d^2} \theta_j X_j\).
The derivatives
\begin{equation}
    \frac{\partial}{\partial \theta_j} \widetilde{\operatorname{cay}}(X(\vec{\theta})) = \left(\operatorname{det}\left(\sqrt[d]{\operatorname{cay}(X(\vec{\theta})}\right)\right)^{-1} \frac{\partial}{\partial \theta_j} \operatorname{cay}(X(\vec{\theta})) +  \left[\frac{\partial}{\partial \theta_j}\left(\operatorname{det}\left(\sqrt[d]{\operatorname{cay}(X(\vec{\theta})}\right)\right)^{-1}\right] \operatorname{cay}(X(\vec{\theta}))
\end{equation}
can be worked out with the Leibniz rule.
Since the \(\frac{\partial}{\partial \theta_j} \operatorname{cay}(X(\vec{\theta}))\) already span the tangent space of \(\mathrm{U}(d)\), which is \(d^2\)-dimensional, and the second term above is always part of the same \(1\)-dimensional subspace spanned by \(\operatorname{cay}(X)\), we get that the span of the \(\frac{\partial}{\partial \theta_j} \widetilde{\operatorname{cay}}(X(\vec{\theta}))\) must be at least \(d^2-1\)-dimensional (cf. e.g. eq. (0.4.5.1) in \cite{johnson_matrix_2013} and set \(\operatorname{rank} B = 1\)).
Hence they must span the tangent space of \(\mathrm{SU}(d)\) and we even get local surjectivity.

\section{Convergence theory on closed subgroups of SU(d)}
\label{app:subgroup-convergence}
As pointed out in the main text, due to the exponential scaling of the Hilbert space dimension, it is natural to also study cases where the ansatz is restricted to a closed subgroup $G \subseteq \mathrm{SU}(d)$ with Lie algebra $\mathfrak g$. 
We want to assume closedness of the subgroup to avoid situations like irrational windings of a torus, where local optima surely become unavoidable. In this appendix we work out under which conditions on $G$, $H$, and $\lvert\psi_0\rangle$ the convergence guarantees of Theorem \ref{thm:su(d)_convergence} survive this restriction.
This will be the proof of Theorem \ref{thm:subgroup_convergence} from the main text.

Write $\rho_0 := \lvert\psi_0\rangle\langle\psi_0\rvert$ and
\begin{equation}
P_0 := i\Big(\rho_0 - \frac{1}{d}\mathds{1}\Big) \in \mathfrak{su}(d)
\end{equation}
for the traceless, skew-Hermitian representative of the initial state (subtracting $i\mathds{1}/d$ only removes the trace of $i\rho_0$ and does not affect the cost function).
As before, $H$ denotes the problem Hamiltonian, taken traceless without loss of generality, so that $iH\in\mathfrak{su}(d)$.
If $G\subseteq \mathrm{SU}(d)$ is a closed subgroup, it is compact, being a closed subgroup of a compact group, and its Lie algebra $\mathfrak g\subseteq\mathfrak{su}(d)$ is therefore reductive,
\begin{equation}
\mathfrak g = \mathfrak s \oplus \mathfrak z(\mathfrak g), \qquad \mathfrak s := [\mathfrak g,\mathfrak g] \ \text{semisimple}, \quad \mathfrak z(\mathfrak g) \ \text{the center of } \mathfrak g.
\end{equation}
We study the cost function \eqref{eq:VQE_cost} with the ansatz restricted to $G$,
\begin{equation}
J\vert_G[U] \;=\; \langle\psi_0\vert U^\dagger H U\vert\psi_0\rangle \;=\; \mathrm{Tr}\{H\, U\rho_0 U^\dagger\}, \qquad U\in G .
\end{equation}
Its Riemannian gradient, left-translated to $\mathfrak g$, is the orthogonal projection of the ambient gradient onto $\mathfrak g$ (see e.g. eq. (36) in \cite{schulte2010gradient}),
\begin{equation}
\label{eq:projected-grad}
\operatorname{grad} J\vert_G[U] = P_{\mathfrak g}\big([H,\, U\rho_0 U^\dagger]\big) = P_{\mathfrak g}\big(-[iH,\,  U P_0 U^\dagger]\big),
\end{equation}
where $P_{\mathfrak g}:\mathfrak{su}(d)\to\mathfrak g$ is the orthogonal projection with respect to the Hilbert--Schmidt inner product.

Since \(g\) is the Lie algebra of a compact Lie group, the adjoint action of \(\mathfrak{g}\) on \(\mathfrak{su}(d)\) is completely reducible.
So $\mathfrak{su}(d)$ splits into an orthogonal (Hilbert--Schmidt) direct sum of $\mathrm{ad}_{\mathfrak{g}}$-invariant (and also \(\mathrm{Ad}_G\)-invariant) subspaces, one of which is $\mathfrak g$ itself. Write accordingly
\begin{equation}
P_0 = P_0^{\mathfrak g} \oplus \tilde P_0, \qquad iH = iH^{\mathfrak{g}} \oplus i \tilde{H}, \qquad P_0^{\mathfrak g}, iH^{\mathfrak{g}} \in \mathfrak g.
\end{equation}
The gradient then decomposes as
\begin{align}
    \operatorname{grad} J\vert_G[U] &= P_{\mathfrak{g}} \big(- \underbrace{[iH^{\mathfrak{g}}, \operatorname{Ad}_U(P_0^{\mathfrak g})]}_{\in \mathfrak{g}} - \underbrace{[iH^{\mathfrak{g}}, \operatorname{Ad}_U(\tilde{P}_0)]}_{\perp \mathfrak{g}} - \underbrace{[i \tilde{H}, \operatorname{Ad}_U(P_0^{\mathfrak g})]}_{\perp \mathfrak{g}} - [i \tilde{H}, \operatorname{Ad}_U(\tilde{P}_0)]\big) \nonumber \\
    \label{eq:gradient_isotypcial_decomposition}
    &= - [iH^{\mathfrak{g}}, \operatorname{Ad}_U(P_0^{\mathfrak g})] - P_{\mathfrak{g}} \big( [i \tilde{H}, \operatorname{Ad}_U(\tilde{P}_0)]\big).
\end{align}
In light of this decomposition, we distinguish three different regimes: one in which both \(iH, P_0 \in \mathfrak{g}\), one in which only \(iH \in \mathfrak{g}\) (the case \(P_0 \in \mathfrak{g}\) works analogously by symmetry of the cost function) and lastly the case where both  \(iH, P_0 \notin \mathfrak{g}\).

In the first two regimes we will show that the the restricted cost functional \(J\vert_G\) still only admits global optima and strict saddles.
Local surjectivity then again guarantees that suboptimal solutions are almost surely avoided and gradient descent either converges to a global optimum, or diverges to infinity.
However, unlike in the full \(\mathrm{SU}(d)\) case, this global optimum need not correspond to a ground state of \(H\) necessarily.

There is an explicit counterexample in the third regime, where local minima appear due to the restriction on the subgroup.
Consider the spin-1 representation of \(\mathfrak{su}(2)\), contained in the defining representation of \(\mathfrak{su}(3)\).
The \(\mathfrak{su}(2)\) generators are given by
\begin{equation}
    i J_x = \frac{1}{\sqrt{2}} \begin{bmatrix}
        0 & i & 0 \\ i & 0 & i \\ 0 & i & 0
    \end{bmatrix}, \qquad i J_y = \frac{1}{\sqrt{2}} \begin{bmatrix}
        0 & 1 & 0 \\ -1 & 0 & 1 \\ 0 & -1 & 0
    \end{bmatrix}, \qquad i J_z = \frac{1}{\sqrt{2}} \begin{bmatrix}
        i & 0 & 0 \\ 0 & 0 & 0 \\ 0 & 0 & -i
    \end{bmatrix},
\end{equation}
in the basis
\begin{equation}
    e_+ = \begin{bmatrix}
        1 \\ 0 \\ 0
    \end{bmatrix}, \qquad e_0 = \begin{bmatrix}
        0 \\ 1 \\ 0
    \end{bmatrix}, \qquad e_- = \begin{bmatrix}
        0 \\ 0 \\ 1
    \end{bmatrix}.
\end{equation}
We choose as the initial state \(\ket{\psi_0} = e_+\) and the Hamiltonian \(H = \frac{1}{2}J_z - 2 J_z^2\).
It is diagonal in our basis with eigenvalues \(E_+ = - \frac{3}{2}\), \(E_0 = 0\) and \(E_- = - \frac{5}{2}\).
Up to an irrelevant phase, the \(\mathrm{SU}(2)\)-orbit of \(\ket{\psi_0}\) is given by \(\ket{\psi(\theta, \phi)} = e^{-i\phi J_z} e^{-i\theta J_y} \ket{\psi_0}\), which evaluates explicitly to
\begin{equation}
    \ket{\psi(\theta, \phi)} = \begin{bmatrix}
        e^{-i\frac{\phi}{\sqrt{2}}} \cos^2(\frac{\theta}{2}) \\
        \frac{1}{\sqrt{2}} \sin(\theta) \\
        e^{i \frac{\phi}{\sqrt{2}}} \sin^2(\frac{\theta}{2}).
    \end{bmatrix}
\end{equation}
In particular, one can see that \(\ket{\psi(0, \phi)} \cong e_+\) and \(\ket{\psi(\pi, \phi)} \cong e_-\).
Therefore the ground state of \(H\) is even contained in the orbit.

At the critical point \(U = \mathds{1}\), the tangent space of \(\mathrm{SU}(2)\) is given by \(\mathfrak{su}(2)\).
We choose \(iJ_x\), \(iJ_y\) and \(iJ_z\) as a basis and can then write the Riemannian Hessian as \(\operatorname{diag}(\frac{3}{2}, \frac{3}{2}, 0)\).
Moreover, the kernel \(\mathds{R}iJ_z\) is tangential to the critical set of the cost function \(J[U]\), since \(\exp(-i\theta J_z)\ket{\psi_0} \cong \ket{\psi_0}\).
Therefore \(U = \mathds{1}\) is a local minimum.

\subsection{Hamiltonian and initial state are in the subalgebra}
\label{app:subgroup-both}

Suppose that $iH\in\mathfrak g$ \emph{and} $P_0\in\mathfrak g$. In this case the restricted VQE can be reformulated entirely on the coadjoint orbit
\begin{equation}
O_{iH} := \{\, U^\dagger(iH)U \mid U\in G \} \cong G/Z_G(iH) \subset \mathfrak g^{*} \cong \mathfrak g,
\end{equation}
where $Z_G(iH)$ is the centralizer of $iH$ in $G$ and we identify $\mathfrak g^{*}\cong\mathfrak g$ via the Hilbert--Schmidt pairing $\langle y,x\rangle=\mathrm{Tr}\{y^\dagger x\}$. The orbit $O_{iH}$ is a, generally partial, flag manifold (partial exactly when $iH$ is not regular in $\mathfrak g$ (see \cite{Duistermaat83}), which is compact and connected because $G$ is, and carries the $G$-invariant Kirillov-Kostant-Souriau symplectic form
\begin{equation}
\Omega_{U^\dagger(iH)U}\big([X,U^\dagger(iH)U],\,[Y,U^\dagger(iH)U]\big) \;=\; -\big\langle iH,\, U[X,Y]U^\dagger\big\rangle, \qquad X,Y\in\mathfrak g.
\end{equation}
The natural (Hamiltonian) $G$-action on $O_{iH}$,
\begin{equation}
\big(V,\, \langle iH, U\cdot U^\dagger\rangle\big) \;\longmapsto\; \langle iH,\, VU\cdot U^\dagger V^\dagger\rangle, \qquad V\in G,
\end{equation}
has the inclusion $\Phi: O_{iH}\hookrightarrow \mathfrak g^{*}$ as its moment map. The restricted cost function is, up to an irrelevant overall sign, the pullback along $\Phi$ of the linear functional on $\mathfrak g^{*}$ associated with $P_0$,
\begin{equation}
\widetilde J_{\lvert\psi_0\rangle} \;=\; \langle \Phi,\, P_0\rangle .
\end{equation}

\begin{proposition}[cf.\ Lemma~32.1 and Theorem~32.6 of \cite{guillemin_symplectic_2001}]
\label{prop:orbit-morse-bott}
Every local minimum of $\widetilde J_{\lvert\psi_0\rangle}$ on $O_{iH}$ is globally minimal, and $\widetilde J_{\lvert\psi_0\rangle}$ is Morse-Bott, so that every saddle point is strict.
\end{proposition}

Consequently, whenever the ansatz $U(\boldsymbol\theta)\in G$ is locally surjective for $\mathfrak g$ in the sense of \cref{def:local_surjectivity}, the argument of \Cref{sec:proof} carries over verbatim with $\mathrm{SU}(d)$ replaced by $G$.
Consequently, gradient descent restricted to $G$ either escapes to infinity or converges almost surely to a global optimizer of $J\vert_G$.
As in the ambient problem, the global optimizer always corresponds to an eigenstate of $H$ (cf. Lemma 1.1 of \cite{Duistermaat83}), but it is a \emph{ground} state of $H$ only if the adjoint $G$-orbit of $\rho_0$ intersects a ground-state projector of $H$. We summarize this as follows.

\begin{corollary}
If $iH\in\mathfrak g$ and $P_0\in\mathfrak g$, and if $U(\vec{\theta})\in G$ is locally surjective for $\mathfrak g$ with uniformly bounded Lie algebra elements $\Omega_j(\boldsymbol\theta)$, then for a sufficiently small step size $\gamma$ and almost all initial points $\vec{\theta}_0$ gradient descent with the ansatz $U(\vec{\theta})$ either diverges to infinity or converges to a global minimum of $J\vert_G$. The global minimum always corresponds to an eigenstate of \(H\). It is a ground state of $H$ if and only if the $G$-orbit of $\ket{\psi_0}$ meets a ground state eigenspace of $H$.
\end{corollary}

\subsection{Hamiltonian from the subalgebra, but arbitrary initial state}
\label{app:subgroup-project}

Assume only that \(iH \in \mathfrak{g}\).
For an arbitrary initial state \(P_0 \notin \mathfrak{g}\) we see that only the first term survives in the decomposition \eqref{eq:gradient_decomposition}, so that a point $U\in G$ is critical if and only if $UP_0^{\mathfrak g}U^\dagger$ commutes with $H$.
In general, \(P_0^{\mathfrak{g}}\) is not even a projector so we cannot hope to obtain convergence to an eigenstate.
However, we can still investigate whether gradient descent converges to global optima on \(G\).

The Hessian of the restricted cost \(J\vert_G\) can be computed (see Proposition III.3 in \cite{schulte2010gradient}) as
\begin{align}
\operatorname{Hess}J\vert_G[U](X,Y) &= \frac{1}{2}\Big(\operatorname{Tr}\{X^\dagger[iH,[Y,UP_0U^\dagger]]\} + \mathrm{Tr}\{X^\dagger[UP_0U^\dagger,[Y,iH]]\}\Big),
\end{align}
for $X,Y\in\mathfrak g$.
Similarly to \cref{eq:gradient_isotypcial_decomposition}, we decompose \(P_0\) into a component \(P_0^\mathfrak{g}\) in \(\mathfrak{g}\) and the orthocomplement \(\tilde{P}_0\)
\begin{align}
    \operatorname{Hess}J\vert_G[U](X,Y) = & \frac{1}{2} \Big( \operatorname{Tr} \{X^\dagger (\underbrace{[iH, [Y, U P_0^\mathfrak{g} U^\dagger]]}_{\in \mathfrak{g}} \oplus \underbrace{[iH, [Y, U \tilde{P}_0 U^\dagger]]}_{\perp \mathfrak{g}})\} \nonumber \\
    & \hphantom{\frac{1}{2} \Big(} + \operatorname{Tr} \{X^\dagger (\underbrace{[U P_0^\mathfrak{g} U^\dagger, [Y, iH]]}_{\in \mathfrak{g}} \oplus \underbrace{[U \tilde{P}_0 U^\dagger, [Y, iH]]}_{\perp \mathfrak{g}})\} \Big)  \\
    = & \frac{1}{2} \Big(\operatorname{Tr}\{X^\dagger[iH,[Y,UP_0^\mathfrak{g}U^\dagger]]\} + \mathrm{Tr}\{X^\dagger[UP_0^\mathfrak{g}U^\dagger,[Y,iH]]\}\Big) = \operatorname{Hess}J_{P_0^\mathfrak{g}}\vert_G[U](X,Y), \nonumber 
\end{align}
where we make with \(J_{P_0^\mathfrak{g}} = \operatorname{Tr}\{H U P_0^\mathfrak{g} U^\dagger\}\) explicit that the cost function only depends on \(P_0^\mathfrak{g} \in \mathfrak{g}\).
According to \Cref{prop:orbit-morse-bott}, the critical-point structure of $J\vert_G$ again consists only of global optima and strict saddles, so we still almost always avoid suboptimal solutions.

\begin{corollary}
    If $iH\in\mathfrak g$ and if $U(\vec{\theta})\in G$ is locally surjective for $\mathfrak g$ with uniformly bounded Lie algebra elements $\Omega_j(\vec{\theta})$, then for a sufficiently small step size $\gamma$ and almost all initial points $\vec{\theta}_0$ gradient descent with the ansatz $U(\vec{\theta})$ either diverges to infinity or converges to a global minimum of $J\vert_G$. This global minimum is a ground state of $H$ if and only if the $G$-orbit of $\ket{\psi_0}$ meets a ground state eigenspace of $H$.
\end{corollary}

As discussed above, due to the projection onto the subalgebra, the critical points of the cost function no longer correspond to eigenstates of \(H\) for an arbitrary initial state.
If we would somehow guarantee that the Riemannian gradient always remains inside of \(\mathfrak{g}\), we would regain this property.
The following section discusses if and how this can occur.

\subsection{Subalgebra forms an ideal in some algebra containing the initial state}
Suppose there is an intermediate subalgebra $\mathfrak g\subseteq\mathfrak h\subseteq\mathfrak{su}(d)$ such that $\mathfrak g$ is an \emph{ideal} in $\mathfrak h$, with
\begin{equation}
iH\in\mathfrak g, \qquad P_0\in\mathfrak h.
\end{equation}
Then, in \cref{eq:projected-grad}, $-iH\in\mathfrak g$ and $iU\rho_0U^\dagger\in\mathfrak h$ (since $\mathfrak h$ is $\mathrm{Ad}_G$-invariant, $\rho_0$ being conjugated only by elements of $G\subset \exp(\mathfrak g)$), so that $[-iH,\,iU\rho_0U^\dagger]\in\mathfrak g$ already, because $\mathfrak g$ is an ideal in $\mathfrak h$.
Hence the projection in \cref{eq:projected-grad} is unnecessary:
\begin{equation}
\operatorname{grad} J\vert_G[U] = [H, U\rho_0U^\dagger] = \mathrm{grad} J[U],
\end{equation}
i.e. the restricted Riemannian gradient coincides \emph{exactly} with the unconstrained one.
The identical reasoning applies to the Hessian.
Consequently every critical point of $J\vert_G$ is automatically an eigenstate of $H$ by the argument of \cref{sec:proof}, and the strict-saddle property transfers directly.

\begin{proposition}
\label{prop:ideal-case}
If $\mathfrak g$ is an ideal in $\mathfrak h\subseteq\mathfrak{su}(d)$ with $iH\in\mathfrak g$ and $P_0\in\mathfrak h$, then the critical points of $J\vert_G$ coincide with the eigenstates of $H$ reachable within $G\ket{\psi_0}$, and $J\vert_G$ has only global optima and strict saddle points.
\end{proposition}
From the same reasoning as in \cref{sec:proof} we immediately obtain the following.
\begin{corollary}
If $iH\in\mathfrak g$ and $P_0\in\mathfrak h$, such that \(\mathfrak{g} \subseteq \mathfrak{h}\) is an ideal and if $U(\boldsymbol\theta)\in G$ is locally surjective for $\mathfrak g$ with uniformly bounded Lie algebra elements $\Omega_j(\vec{\theta})$, then for a sufficiently small step size $\gamma$ and almost all initial points $\vec{\theta}_0$ gradient descent with the ansatz $U(\vec{\theta})$ either diverges to infinity or converges to a global minimum of $J\vert_G$. The global minimum always corresponds to an eigenstate of \(H\). It is a ground state of $H$ if and only if the $G$-orbit of $\ket{\psi_0}$ meets a ground state eigenspace of $H$.
\end{corollary}

In the following we characterize the subalgebras $\mathfrak h\supseteq\mathfrak g$ for which Proposition~\ref{prop:ideal-case} applies.

\begin{lemma}
\label{lem:minimal-h}
Without loss of generality one may take $\mathfrak h = \langle \mathfrak g \cup \{P_0\}\rangle_{\mathrm{Lie}}$, the Lie algebra generated by $\mathfrak g$ and $P_0$.
\end{lemma}
\begin{proof}
Let $\mathfrak c := \langle \mathfrak g \cup \{P_0\}\rangle_{\mathrm{Lie}}$. If $\mathfrak g\subseteq\mathfrak h\subseteq\mathfrak{su}(d)$ with $\mathfrak g$ an ideal in $\mathfrak h$ and $P_0\in\mathfrak h$, then $\mathfrak g\subseteq\mathfrak c\subseteq\mathfrak h$, and since $[\mathfrak g,\mathfrak c]\subseteq[\mathfrak g,\mathfrak h]\subseteq\mathfrak g$, $\mathfrak g$ is already an ideal in the (possibly smaller) subalgebra $\mathfrak c$.
\end{proof}

By Lemma~\ref{lem:minimal-h} it suffices to characterize, for a given $X\in\mathfrak{su}(d)\setminus\mathfrak g$, when $\mathfrak g$ is an ideal in $\langle \mathfrak g\cup\{X\}\rangle_{\mathrm{Lie}}$ (equivalently, when $\operatorname{grad} J\vert_G = \operatorname{grad} J$ at every point).

\begin{lemma}
\label{lem:necessary-sufficient-ideal}
$\mathfrak g$ is an ideal in $\langle \mathfrak g \cup \{X\}\rangle_{\mathrm{Lie}}$ if and only if $[\mathfrak g,X]\subseteq\mathfrak g$. In that case $\dim \langle \mathfrak g\cup\{X\}\rangle_{\mathrm{Lie}} \le \dim\mathfrak g + 1$.
\end{lemma}
\begin{proof}
Necessity is immediate. For sufficiency, note that every element of $\langle \mathfrak g\cup\{X\}\rangle_{\mathrm{Lie}}$ is a linear combination of nested commutators of elements of $\mathfrak g$ and $X$; if $[\mathfrak g,X]\subseteq\mathfrak g$, any such nested commutator involving $X$ at least once lies in $\mathfrak g$ (all further brackets with $\mathfrak g$-elements stay inside $\mathfrak g$), while brackets of $\mathfrak g$ with itself trivially stay in $\mathfrak g$. Hence $\langle \mathfrak g\cup\{X\}\rangle_{\mathrm{Lie}} = \mathfrak g \oplus \mathbb R X$ as a vector space (with $\mathfrak g$ an ideal), which proves both claims.
\end{proof}

Since $\mathfrak g$ is compact it is reductive, $\mathfrak g=\mathfrak s\oplus\mathfrak z(\mathfrak g)$. Define the linear map $D:=\mathrm{ad}_X\vert_{\mathfrak g}:\mathfrak g\to\mathfrak g$, which is a well-defined derivation of $\mathfrak g$ under the hypothesis $[\mathfrak g,X]\subseteq\mathfrak g$ of Lemma~\ref{lem:necessary-sufficient-ideal}.

\begin{lemma}
\label{lem:D-respects-splitting}
$D(\mathfrak s)\subseteq\mathfrak s$ and $D(\mathfrak z(\mathfrak g))\subseteq\mathfrak z(\mathfrak g)$.
\end{lemma}
\begin{proof}
For $s_1=[g_1,g_2]\in\mathfrak s$ with $g_1,g_2\in\mathfrak g$, the Leibniz rule gives $D(s_1)=[D(g_1),g_2]+[g_1,D(g_2)]\in\mathfrak s$, since $\mathfrak s=[\mathfrak g,\mathfrak g]$ is itself an ideal in $\mathfrak g$. For $z\in\mathfrak z(\mathfrak g)$ and any $g\in\mathfrak g$, $[z,g]=0$ and $[z,D(g)]=0$ (as $D(g)\in\mathfrak g$ and $z$ is central), so
\begin{equation}
[D(z),g] \;=\; [D(z),g]+[z,D(g)] \;=\; D([z,g]) \;=\; D(0) \;=\; 0,
\end{equation}
i.e.\ $D(z)$ centralizes $\mathfrak g$, so $D(z)\in\mathfrak z(\mathfrak g)$.
\end{proof}

By Lemma~\ref{lem:D-respects-splitting}, $D\vert_{\mathfrak s}$ is a derivation of the semisimple algebra $\mathfrak s$, hence by the first Whitehead Lemma (Lemma~7.5.26 of \cite{hilgert_structure_2012}) it is inner, i.e. $D\vert_{\mathfrak s}=\mathrm{ad}_Y$ for some $Y\in\mathfrak s\subseteq\mathfrak g$. Replacing $X$ by $X':=X-Y$, which generates the same $\mathfrak h$, yields $D':=\mathrm{ad}_{X'}\vert_{\mathfrak g}$ annihilating $\mathfrak s$.

\begin{lemma}
\label{lem:X-centralizes}
$D'$ also annihilates $\mathfrak z(\mathfrak g)$, so $X'=X-Y$ centralizes all of $\mathfrak g$.
\end{lemma}
\begin{proof}
If \(\mathfrak{z}(\mathfrak{g}) = \{0\}\), the statement is immediately true, so w.l.o.g. \(\mathfrak{z}(\mathfrak{g}) \neq \{0\}\).
If \(D'\) did not annihilate $\mathfrak z(\mathfrak g)$, the subalgebra $\mathfrak z(\mathfrak g)\oplus\mathbb R X'$ would satisfy $[\mathfrak z(\mathfrak g)\oplus\mathbb R X', \mathfrak z(\mathfrak g)\oplus\mathbb R X'] = \mathfrak{c} \subseteq \mathfrak z(\mathfrak g)$, \(\mathfrak{c} \neq \{0\}\).
But since \(\mathfrak{c} \subseteq \mathfrak z(\mathfrak g)\) we have that $[\mathfrak c,\mathfrak c]=0$, i.e. $\mathfrak z(\mathfrak g)\oplus\mathbb R X'$ would be solvable but non-abelian.
As a subalgebra of the compact algebra $\mathfrak{su}(d)$ it would also have to be reductive, but reductive and solvable implies abelian, which yields a contradiction.
\end{proof}

Combining \Cref{lem:D-respects-splitting,lem:X-centralizes}, up to addition of an element of $\mathfrak s\subseteq\mathfrak g$, $X$ can always be chosen inside the centralizer $Z_{\mathfrak{su}(d)}(\mathfrak g)$, in which case $\mathfrak h\cong\mathfrak g\oplus\mathfrak u(1)$. It remains to determine when $Z_{\mathfrak{su}(d)}(\mathfrak g)$ contains an element outside $\mathfrak g$, i.e.\ when $\mathfrak z(\mathfrak g)\not\subseteq Z_{\mathfrak{su}(d)}(\mathfrak g)$.

By Schur's Lemma, decomposing $\mathbb C^d$ into $\mathfrak g$-isotypic components
\begin{equation}
\mathds{C}^d \cong \bigoplus_{i=1}^{k} V_i \otimes \mathds{C}^{m_i},
\end{equation}
where $V_i$ pairwise inequivalent irreps and $m_i$ their multiplicities, the centralizer is explicitly given by
\begin{equation}
Z_{\mathfrak{su}(d)}(\mathfrak g) = \Big[\bigoplus_{i=1}^k \mathds{1}_{V_i}\otimes\mathfrak u(m_i)\Big] \cap \mathfrak{su}(d).
\end{equation}
On the other hand, again by Schur's Lemma the center acts as a scalar on each isotypic block
\begin{equation}
\mathfrak z(\mathfrak g) \subseteq \Big\{\bigoplus_{i=1}^{k} \lambda_i\, \mathds{1}_{V_i}\otimes \mathbb \mathds{1}_{m_i}\Big\} \cap \mathfrak{su}(d).
\end{equation}
This explicit representation yields immediately the characterization in the following proposition.
\begin{proposition}
\label{prop:centralizer-characterization}
$\mathfrak z(\mathfrak g)\subsetneq Z_{\mathfrak{su}(d)}(\mathfrak g)$ (equivalently, a suitable $X\notin\mathfrak g$ with $[\mathfrak g,X]\subseteq\mathfrak g$ exists) if and only if
\begin{equation}
\exists\, i:\ m_i>1 \qquad\text{or}\qquad \dim\mathfrak z(\mathfrak g) < k-1 .
\end{equation}
\end{proposition}
\Cref{prop:centralizer-characterization} completely characterizes the Lie algebras \(\mathfrak{g}\) that admit \emph{some} enlargement $\mathfrak h=\mathfrak g\oplus\mathfrak u(1)$ where $\mathfrak g$ is an ideal. 
It does not yet guarantee that the specific coset $X+\mathfrak g$ contains an operator of the required form $P_0=i(\rho_0-\mathds{1}/d)$ for a genuine pure state $\lvert\psi_0\rangle\in\mathbb C^d$.
We leave this question open for future work.

\subsection{Ground-state reachability and quasi-minuscule representations}
\label{app:subgroup-reachability}

While the results from above guarantee that gradient descent can only converge to a global minimum on the closed subgroup (or not converge at all), this minimum need no longer necessarily correspond to a ground state of \(H\).
It only corresponds to a ground state exactly if the \(G\)-orbit of \(\ket{\psi_0}\) contains a ground state.
The question of ground state reachability on subgroup orbits was recently studied by \cite{singh_ground-state_2026}.
They identify as a necessary but not sufficient condition that the initial state and a ground state of \(H\) have support in the same isotypic components under the action of \(G\).

Here we are going to provide a sufficient (but not necessarily equivalent) condition for a ground state to be in the \(G\)-orbit of \(\ket{\psi_0}\) and prove Theorem \ref{thm:quasi-minuscule} from the main text.
It is formulated in the language of the representation theory of compact Lie groups, which we assume here to be known.
An introduction to the topic for the unfamiliar reader can be found in Chapters 7-12 of \cite{Hall15}.
We begin by defining the notion of a minuscule representation (see e.g. Exercise 24 in Chapter 6.1 of \cite{Bourb08a}).

\begin{definition}
An irreducible representation is \emph{minuscule} if the associated Weyl group $W$ acts transitively on its weights, i.e.\ for any two weights $\lambda,\mu$ there is $w\in W$ with $\lambda = w\cdot\mu$ or equivalently, all weights lie in a single Weyl-group orbit.
\end{definition}

A natural generalization which turns out to be better suited for our purposes is given by so-called quasi-minuscule representations (see e.g. Chapter 2 of \cite{lee_polytopes_2017}).

\begin{definition}
An irreducible representation is \emph{quasi-minuscule} if all of its \emph{non-zero} weights lie in a single Weyl-group orbit.
\end{definition}

The (quasi-)minuscule representations of the simple Lie algebras are completely categorized \cite{lee_polytopes_2017}.

Let $G$ be a compact Lie group with Lie algebra $\mathfrak g$ and let $G$ be represented on $\mathds{C}^d$ by a quasi-minuscule representation.
Then \(\mathfrak{g}\) is the compact real form of its complexification $\mathfrak g_{\mathds{C}}$ and the given representation extends to an irreducible quasi-minuscule representation of $\mathfrak g_{\mathds{C}}$.
Fix a maximal commutative subalgebra $\mathfrak t\subset\mathfrak g$ (equivalently, a maximal torus $T \subset G$) such that the representation is of highest weight with respect to the Cartan subalgebra $\mathfrak t_{\mathds{C}}=\mathfrak t+i\mathfrak t\subset\mathfrak g_{\mathds{C}}$.
Then $iH$ is $G$-conjugate to $i\tilde{H}\in\mathfrak t$, with eigenvalues of $\tilde{H}$ related to the weights $\lambda$ of the representation via $\tilde{H} v = i\langle\lambda,\tilde{H}\rangle v$ for some weight vector $v$ associated to the weight \(\lambda\).

If at least one non-zero weight of $\tilde{H}$ corresponds to a ground state, then there exists a Weyl-group element $w\in W$ mapping the highest weight $\mu$ to a ground state weight $\lambda=w\cdot\mu$ of $\tilde{H}$ (and correspondingly for the weight vectors), so that starting from the highest-weight vector $\lvert\psi_0\rangle\in V_\mu$ there exists a ground state of $H$ in the $G$-orbit of $\lvert\psi_0\rangle$.
The assumption that some non-zero weight of $\tilde{H}$ be a ground-state weight can always be fulfilled except in trivial cases.
Quasi-minuscule representations that are not minuscule have roots as their highest weights (\cite{stembridge_computational_2001} p. 30), so every non-zero weight occurs in $\pm\lambda$ pairs, and a non-zero $\tilde{H}$ can therefore never have the zero-weight space as its unique ground state eigenspace.
This proves Theorem \ref{thm:quasi-minuscule} from the main text.

Incorporating the results of the previous subsections on the convergence of the gradient descent, we immediately obtain the following corollary.
\begin{corollary}
    \label{cor:quasi-minuscule}
    Let $iH\in\mathfrak g$ and $U(\vec{\theta})\in G$ be locally surjective for $\mathfrak g$ with uniformly bounded Lie algebra elements $\Omega_j(\vec{\theta})$. Further let the representation of \(\mathfrak{g}\) on \(\mathds{C}^d\) extend to a quasi-minuscule representation of its complexification \(\mathfrak{g}_{\mathds{C}}\).
    Take \(\ket{\psi_0}\) to be the highest-weight vector.
    Then for a sufficiently small step size $\gamma$ and almost all initial points $\vec{\theta}_0$ gradient descent with the ansatz $U(\vec{\theta})$ either diverges to infinity or converges to a ground state of \(H\).
\end{corollary}

\end{document}